\documentclass[a4paper,11pt]{article}
\usepackage[top=1in,bottom=1in,left=1in,right=1in]{geometry}
\usepackage{amsmath,amssymb,amsthm,times,graphics,graphicx,bm}
\usepackage[mathscr]{euscript}
\usepackage{verbatim}
\usepackage[usenames,dvipsnames]{color}
\usepackage{paralist}
\usepackage{hyperref}
\usepackage{mathtools}

\hypersetup{
	pdftitle={A linear time algorithm for quantum $2$-SAT},
	pdfauthor={N. de Beaudrap and S. Gharibian},
	colorlinks,%
	citecolor=black,%
	filecolor=black,%
	linkcolor=black,%
	urlcolor=black
}

\usepackage{enumitem}
\usepackage{framed}
\usepackage{hyperref}

\newcommand\ie{\emph{i.e.}}
\newcommand\eg{\emph{e.g.}}
\newcommand\etc{\emph{etc.}}
\newcommand\etal{\mbox{\itshape et al.}}

\newcommand\cH{\mathcal H}

\newcommand\xT{\mathsf T}

\newcommand\parStyle[1]{\textrm{\mdseries\upshape({#1}\kern0.1ex)}}

\newlength\romanumlabelwd
\makeatletter
\def\tagform@#1{%
	\ifmmode
		\mbox{\normalsize\maketag@@@{(\ignorespaces#1\unskip\@@italiccorr)}}%
	\else
		\maketag@@@{(\ignorespaces#1\unskip\@@italiccorr)}%
	\fi}
\makeatother

\newtheoremstyle{theorem?}
  {\topsep}{\topsep}   																	
  {\itshape}{0pt}{\bfseries}{?}{5pt plus 1pt minus 1pt} 
  {}          																					
\theoremstyle{theorem?}

\makeatletter
\def\@cludgescbf#1#2\end{\textbf{#1\scalefont{0.85}#2}}
\newcommand\cludgescbf[1]{\@cludgescbf#1\end}
\makeatother

\newtheoremstyle{Axiom}
  {\topsep}{\topsep}   																	
  {\itshape}{0pt}{\bfseries}{?}{5pt plus 1pt minus 1pt} 
  {}          																					
\theoremstyle{Axiom}

\theoremstyle{definition}
\newtheorem{definition}{Definition}

\newtheorem*{definition*}{Definition}

\theoremstyle{plain}
\newtheorem{theorem}{Theorem}
\newtheorem{lemma}[theorem]{Lemma}

\newtheorem*{proposition*}{Proposition}

\newtheorem*{claim*}{Claim}
\newtheorem{corollary}{Corollary}[theorem]

\DeclareMathOperator\Tr{Tr}

\newcommand\D{\mathbb{D}}
\newcommand\E{\mathbb{E}}
\newcommand\F{\mathbb{F}}
\newcommand\C{\mathbb{C}}

\newcommand\Z{\mathbb{Z}}
\newcommand\Q{\mathbb{Q}}

\renewcommand\vec\mathbf

\newcommand\ox{\otimes}
\newcommand\x{\times}


\let\subset\subseteq
\let\oldepsilon\epsilon
\let\epsilon\varepsilon
\let\varepsilon\oldepsilon
\let\le\leqslant
\let\ge\geqslant

\def\href#1#2{\texttt{#2}}

\makeatletter
	\newcommand\ket[1]{\left| #1 \right\rangle\@ifnextchar\bra{\mspace{-4mu}}{}}
	\newcommand\bra[1]{\left\langle #1 \right|}

	\newcommand\bracket[2]{\left\langle #1 \left| #2 \right\rangle\right.}
\makeatother

\newcommand{\lin}[1]{{\rm L}\left(#1\right)}

\newcommand{\spa}[1]{\mathcal{#1}}
\newcommand{\complex}{{\mathbb C}}

\newcommand{\ketbra}[2]{\ket{#1}\!\bra{#2}}        
\newcommand{\set}[1]{{\left\{#1\right\}}}    
\newcommand{\tq}{\mbox{2-QSAT}}

\newcommand{\solveq}{\mbox{SOLVE$_{\mathrm Q}$}}

\newcommand\rev{^{\mathsf R}}

\makeatletter
\renewcommand\labelenumii{\expandafter\@gobble\theenumii.}
\renewcommand\theenumii{\theenumi\alph{enumii}}
\renewcommand\labelenumiii{\expandafter\@gobble\theenumiii}
\renewcommand\theenumiii{\theenumii(\textit{\rmfamily\roman{enumiii}}\:\!)}
\makeatother

\begin{document}

\title{A linear time algorithm for quantum $2$-SAT}
\author{Niel de Beaudrap\footnote{Department of Computer Science, University of Oxford, Oxford UK, OX1 3QD. Email: niel.debeaudrap@cs.ox.ac.uk.} \and Sevag Gharibian\footnote{Department of Computer Science, Virginia Commonwealth University, Richmond, VA 23284, USA. Email: sgharibian@vcu.edu.}}

\date{}
\maketitle
\begin{abstract}
	The Boolean constraint satisfaction problem $3$-SAT is arguably the canonical NP-complete problem. In contrast, $2$-SAT can not only be decided in polynomial time, but in fact in deterministic linear time. In 2006, Bravyi proposed a physically motivated generalization of $k$-SAT to the quantum setting, {defining the problem ``quantum $k$-SAT''}. He showed that quantum $2$-SAT is also solvable in polynomial time on a classical computer, in particular in deterministic time $O(n^4)$, {assuming unit-cost arithmetic over a field extension of the rational numbers,} where $n$ is number of variables.
  In this paper, we present an algorithm for quantum $2$-SAT which runs in linear time, \ie~deterministic time $O(n+m)$ for $n$ and $m$ the number of variables and clauses, respectively. Our approach exploits the transfer matrix techniques of Laumann \emph{et al.} [QIC, 2010] used in the study of phase transitions for random quantum $2$-SAT, and bears similarities with both the linear time $2$-SAT algorithms of Even, Itai, and Shamir (based on backtracking) [SICOMP, 1976] and Aspvall, Plass, and Tarjan (based on strongly connected components) [IPL, 1979].
\end{abstract}

\section{Introduction}\label{scn:intro}
Boolean constraint satisfaction problems lie at the heart of theoretical computer science. Among the most fundamental of these is $k$-SAT, in which one is given a formula $\phi$ on $n$ variables, consisting of a conjunction $\phi(x) = C_1 \wedge C_2 \wedge \cdots \wedge C_m$ of $m$ clauses, each of which is a disjunction of $k$ literals, \eg~$(x_h \vee \bar x_i \vee x_j)$ for $1 \le h,i,j \le n$.
The problem is to determine whether there exists an assignment $x\in\set{0,1}^n$ which simultaneously satisfies all of the constraints $C_i$, \ie~for which $\phi(x)=1$.
While \mbox{$3$-SAT} is NP-complete~\cite{C72,L73,K72}, $2$-SAT admits a number of polynomial time algorithms~(\eg~\cite{DP60,K67,EIS76,APT79,P91}), the fastest of which require just linear time~\cite{EIS76,APT79}.

In 2006, Bravyi~\cite{B06} introduced \mbox{$k$-QSAT}, a problem which generalizes
 \mbox{$k$-SAT}, {as follows.
In place of clauses $C_i$, acting on $k$-bit substrings of $n$ bit strings $x \in \{0,1\}^n$, one considers orthogonal projectors $\bar\Pi_{i}$ which act on $k$-qubit subsystems of an $n$-qubit system $\ket{\psi} \in \cH^{\otimes n}$, where $\cH := \C^2$.
(A sketch of how $k$-SAT can be embedded into $k$-QSAT is given in Section~\ref{scn:preliminaries}.)
These projectors extend to act on states $\ket{\psi}$ by defining $\Pi_i = \bar\Pi_i \ox I$, so that $\Pi_i$ acts as the identity on all tensor factors apart from those qubits on which $\bar\Pi_i$ is defined.
One then considers $\ket{\psi}$ to ``satisfy'' the \tq\ instance if $\Pi_i \ket{\psi} = 0$ for all $i$.}
This formulation may be motivated, \eg,~by problems in many-body physics~\cite{dBOE10,CCDJZ11}.
While $3$-QSAT is complete for QMA$_1$~\cite{B06,GN13} (a quantum generalization\footnote{Here, Quantum Merlin Arthur (QMA) is the quantum analogue of Merlin-Arthur (MA) in which the proof and verifier are quantum, and QMA$_1$ is QMA with perfect completeness. Unlike the classical setting, in which MA is known to admit perfect completeness~\cite{ZF87,GZ11}, whether QMA$=$QMA$_1$ remains open (see e.g.~\cite{JKNN12}).} of NP), \tq\  is solvable in deterministic polynomial time~\cite{B06}, using $O(n^4)$ field operations over $\complex$. 

Given the existence of linear time algorithms for classical $2$-SAT, this raises the natural question:
    \emph{Can \tq\  also be solved in linear time?}
Our main result in this paper is as follows.

\begin{theorem}\label{thm:main}
    There exists a deterministic algorithm \solveq\ which, given an instance of \tq, outputs a representation of a satisfying assignment if one exists (presented as a list of one- and two-qubit unit vectors to be taken as a tensor product), and rejects otherwise.
		The algorithm halts in time $O(n+m)$ on inputs on $n$ qubits with $m$ projectors (assuming unit-cost operations over $\C$).
    Furthermore:
    \begin{itemize}
		\item
			\solveq\ can produce its output using ${O((n+m) M(n))}$ bit operations, where $M(n)$ is the asymptotic upper bound on the cost of multiplying two $n$ bit numbers;
		\item
			If the projectors are all product projectors, the algorithm \solveq\ requires only $O(n+m)$ bit operations regardless of what computable subfield $\F \subset \C$ the projector coefficients range over.
    \end{itemize}
\end{theorem}

\noindent\emph{Remarks:} {The setting of product constraints above includes classical $2$-SAT (see Section~\ref{scn:preliminaries}): in this case the bit-complexity of our algorithm matches optimal $2$-SAT algorithms~\cite{EIS76,APT79}.
For more general instances of \tq, the $O((m+n) M(n))$ bit-complexity of our algorithm is compares favourably to the complexity of extracting a satisfying assignment using Bravyi's \tq\ algorithm, which requires $O(n^4 M(n))$ bit operations if one uses similar algebraic algorithms to ours.
In ``Significance and open questions'' below, we discuss the question of field-operation-complexity vs.\ bit-complexity, as well as whether our algorithm is tight in terms of bit complexity.}


\paragraph{Techniques employed.} The origin of this work is the observation that Bravyi's \tq\  algorithm can be thought of as an analogue of Krom's $2$-SAT algorithm~\cite{K67}, which involves computing the transitive closure of directed graphs. Krom's algorithm repeatedly applies a fixed inference rule for each pair of clauses sharing a variable. The repeated application of the inference rule leads to an $O(n^3)$ time to determine satisfiability and an $O(n^4)$ time to compute a satisfying assignment. Bravyi's algorithm has the same runtimes, measured in terms of the number of field operations.

This work aims to develop a quantum analogue of Aspvall, Plass, and Tarjan's (APT) linear time \mbox{$2$-SAT} algorithm~\cite{APT79}, which reduces $2$-SAT to computing the strongly connected components of a directed graph.
Note that classically $(\alpha \vee \beta)$ is equivalent to $(\bar \alpha \Rightarrow \beta)$ and $(\bar \beta \Rightarrow \alpha)$, for literals $\alpha$ and $\beta$.
APT constructs an \emph{implication graph} $G$ of a $2$-SAT instance $\phi$, with vertices labelled by literals $x_i$ and $\bar x_i$ for each $i$, and edges ${\bar \alpha \to \beta}$ and ${\bar \beta \to \alpha}$ for each clause $(\alpha\vee \beta)$.
Then, they show that $\phi$ is satisfiable if and only if $x_i$ and $\bar x_i$ are not in the same strongly connected component of $G$ for any $i$~\cite{APT79}. As the strongly connected components of $G$ can be computed in linear time~\cite{T72}, this yields a linear time algorithm for $2$-SAT.

In the quantum setting, not all $n$-qubit states can be described by assignments to individual qubits (\eg,~entangled states).
Fortunately, Chen~\etal~\cite{CCDJZ11} show that we may reduce any instance of \tq\ to an instance which is satisfiable if and only if there is a satisfying state, in which qubits have separate assignments (see Section~\ref{scn:preliminaries} for details).
In this setting, there is a natural analogue of the equivalence ${(x_i\vee x_j)}\equiv {(\bar x_i\Rightarrow x_j)}\wedge{(\bar x_j\Rightarrow x_i)}$ in terms of so-called ``transfer matrices'' (e.g.~\cite{B06,LMSS10}). For any rank-$1$ quantum constraint $\Pi_{ij}\in \lin{\complex^2\otimes \complex^2}$ on qubits $i$ and $j$, there exists a \emph{transfer matrix} $\xT_{ij}\in \lin{\complex^2}$, such that for any assignment $\ket{\psi_i}$ to qubit $i$ such that $\xT_{ij}\ket{\psi_i}\neq 0$, the state on qubit $j$ for which the constraint $\Pi_{ij}$ is satisfied is given by $\xT_{ij}\ket{\psi_i}$.\footnote{%
	The usual convention is to describe quantum states by unit vectors in $\C^2$, albeit up to equivalence under multiplication by $z \in \C$ for $|z| = 1$.
	However, vectors produced via transfer matrices might not be normalised.
	As we are not explicitly concerned with the probabilities of any measurement outcomes obtained from quantum processes, we represent quantum states by vectors which are equivalent up to multiplication by arbitrary scalar factors.}		
(Conversely, for any $\xT_{ij} \in \lin{\C^2}$, there is a unique rank-1 orthogonal projector $\Pi_{ij} \in \lin{\C^2 \ox \C^2}$ whose nullspace is spanned by $\ket{\psi_i} \ox \xT_{ij}\ket{\psi_i}$ for $\ket{\psi_i}$ ranging over $\C^2$.)
This suggests a quantum analogue $G$ of an implication graph: For each possible assignment $\ket{\psi}$ to a qubit $i$, we define a vertex ${(i,\ket{\psi})}$, and include a directed edge ${(i,\ket{\psi})} \to {(j,\ket{\phi})}$ if there is a transfer matrix $\xT_{ij}$ (corresponding to some constraint $\Pi_{ij}$) such that $\xT_{ij}\ket{\psi}=c\ket{\phi}$ for some $c \ne 0$.
We then ask if for each qubit $i$, there is a vertex $(i,\ket{\psi_i})$ which cannot reach any $(i,\ket{\psi'_i)}$ where $\ket{\psi_i} \not\propto \ket{\psi'_i}$.
If there are such paths $(i,\ket{\psi_i}) \to \cdots \to (i,\ket{\psi'_i)}$ for all $\ket{\psi_i}$, this is analogous to $x_i$ and $\bar x_i$ being in a common strong component in the APT algorithm.

As it stands, this approach has a shortcoming: In the quantum regime, each qubit has a \emph{continuum} of possible assignments (rather than two), which may generate unbounded orbits in an APT-style algorithm.
However, by applying techniques of Laumann \emph{et al.}~\cite{LMSS10} from the study of phase transitions in random \tq, we may in some cases reduce the set of possible assignments for a qubit $i$ to one or two.
Consider the \emph{interaction graph} $G'$ of a \tq\ instance, in which vertices correspond to qubits, and two vertices are connected by an (undirected) edge if the corresponding qubits $i$ and $j$ are subject to a constraint $\Pi_{ij}$. Suppose $C=(v_1,\ldots,v_t,v_1)$ is a cycle in $G'$, with transfer matrices $\xT_{v_iv_{i+1}}$ arising from each constraint $\Pi_{v_i v_{i+1}}$\,, and compute $\xT_C:=\xT_{v_tv_1}\cdots \xT_{v_2v_3}\xT_{v_1v_2}$. If $\xT_C$ has a non-degenerate spectrum, then the only possible satisfying assignments for $v_1$ are eigenvectors of $\xT_C$~\cite{LMSS10} (see also Lemma~\ref{l:inconsistency}).
In effect, computing $\xT_C$ ``simulates'' uncountably many (!) traversals ${(i,\ket{\alpha})} \to \cdots \to {(i,\ket{\beta})}$ in $G$; restricting to the eigenvectors of $\xT_C$ corresponds to ignoring vertices in $G$ which are infinitely far from the top of any topological order of $G$.
If we hence {describe} cycles $C$ with non-degenerate $\xT_C$ as \emph{discretizing}, this suggests the approach of finding a discretizing cycle at each qubit $i$, and using it to reduce the number of possible states on $i$ to one or two.
This simple principle is the starting point of our work.

Despite this simplicity, some obstacles must be addressed to obtain a linear-time algorithm.
In the setting of random \tq~\cite{LMSS10}, {every cycle $C$ is a discretising cycle with probability one, as there is zero probability that either a transfer matrix is singular, or that a product of them has a degenerate spectrum.
This allows one to quickly reduce the space of assignments possible for a qubit.
In contrast, in {our setting (\ie,~worst case analysis)}, we cannot assume such a distribution of transfer matrices arising from a \tq\ instance.
For instance, any constraint $\Pi_{ij}$ corresponding to a product operator (\eg, a classical \mbox{2-SAT} constraint) has a singular transfer matrix, which when multiplied with other singular matrices may give rise to a singular cycle matrix.}
Even if a discretising cycle $C$ does exist using some of the edges $jk$, $k\ell$, \ldots, we may have to traverse those edges multiple times to discover $C$, which is worrisome for a linear-time algorithm. 
Furthermore, we must address the case in which there are no discretising cycles at all to get a discrete algorithm started.
In order to demonstrate a \emph{linear}-time algorithm for \tq\ in the spirit of APT, these problems must be carefully addressed.

Our approach to resolve these issues is as {follows.  In} an instance of \tq\ in which all transfer matrices are non-singular, we show that discretising cycles are easy to find if they exist, and that the absence of discretising cycles allows one to easily obtain a satisfying state.
If, on the other hand, singular transfer matrices are {present, the} corresponding product constraints $\Pi_{ij} = \ketbra{\alpha}{\alpha}_i \ox \ketbra{\beta}{\beta}_j$ \emph{themselves} impose a different discretising {influence: If} $\ket{\smash{\alpha^\bot}}$ and $\ket{\smash{\beta^\bot}}$ are states orthogonal to $\ket{\alpha}$ and $\ket{\beta}$ respectively, then at least one of the assignments $(i,\ket{\smash{\alpha^\bot}})$ or $(j,\ket{\smash{\beta^\bot}})$ is required for a satisfying assignment.
This leads us to adopt an approach of ``trial assignments'' which is highly reminiscent of another linear-time \mbox{2-SAT} algorithm due to Even-Itai-Shamir~\cite{EIS76}, which attempts to reduce to an instance of \tq\ with fewer product constraints by determining partial assignments satisfying $\Pi_{ij}$.
(For simplicity, we also adopt the approach of trial assignments for qubits whose state-space have been reduced by discretizing cycles.)
{This leads us to our} algorithm \solveq\ {(Figure~\ref{fig:algorithmq}, in Section~\ref{scn:alg})}, which combines elements of both the Even-Itai-Shamir~\cite{EIS76} and Apsvall-Plass-Tarjan~\cite{APT79} linear-time $2$-SAT algorithms as described above.


\paragraph{Previous work.} There is a long history of polynomial time solutions for classical $2$-SAT~\cite{Q59, DP60,K67,EIS76,APT79,P91}, ranging from time $O(n^4)$ to $O(n+m)$. As we mention above, the most relevant of these to our setting are the algorithms of Even, Itai, and Shamir~\cite{EIS76} (based on limited backtracking) and Aspvall, Plass, Tarjan~\cite{APT79} (based on strongly connected component detection).

In contrast, little work has been performed in the quantum setting.
Until recently, Bravyi's algorithm was the only explicitly articulated algorithm for \tq, and requires $O(n^4)$ field operations and $O(n^4M(n))$ bit operations. Other work on \tq\ instead concerns either the structure of the solution space of instances of \tq~\cite{LMSS10,dBOE10,CCDJZ11}, or bounds on counting complexity~\cite{JWZ11,Beaudrap14}.


Propagation of assignments using transfer matrices is present already in Bravyi~\cite{B06}, and the results of Laumann~\etal~\cite{LMSS10} allow us to restrict the possibly satisfying states on single qubits by finding discretising cycles.
We incorporate these into efficient discrete algorithms for testing possible assignments, and provide a cost analysis in terms of field operations and bit operations. In contrast to the random \tq\ setting of~\cite{LMSS10}, we do not assume any particular distribution on constraints.

\emph{Note:} Very recently, Arad~\etal~\cite{ASSZ-2015} independently presented an algorithm for \tq, which also runs in $O(n+m)$ time  using unit-cost field operations.
The overall structure of our algorithm appears similar to theirs, though our treatment of the key issue of \tq~instances with only entangled constraints appears to use different techniques (in particular, Ref.~\cite{ASSZ-2015} appears to be based on results of Ji, Wei, Zeng~\cite{JWZ11} which modify the instance itself, whereas we use ideas of~\cite{LMSS10} to tackle the existing instance via the concept of discretizing cycles). As well as obtaining an upper bound on field operations matching Ref.~\cite{ASSZ-2015}, we also include an analysis of the bit complexity of our algorithm \solveq, and in particular indicate how our algorithm matches the asymptotic bit complexity of the best algorithms on classical instances of 2-SAT.


\paragraph{Significance and open questions.}
\label{para:significance}%
Quantum $k$-SAT and its optimization variant, $k$-LOCAL HAMILTONIAN~\cite{KSV02}, are central to quantum complexity theory, being canonical QMA$_1$- and QMA-complete problems for $k \ge 3$ and $k \ge 2$ respectively~\cite{B06,GN13,KSV02,KKR06}, just as \mbox{$k$-SAT} and \mbox{MAX-$k$-SAT} are central to classical complexity theory as canonical NP-complete problems.
Moreover, quantum $k$-SAT may be motivated {using the notion of \emph{frustration-freeness} in many-body physics~\cite{H06,BT09}.} It is thus natural to ask what the minimal resources to solve quantum $k$-SAT are.

{We now discuss} the number of field operations used by our algorithm, $O(m+n)$, versus the number of bit operations, $O((m+n)M(n))$, in Theorem~\ref{thm:main}. There is no such distinction in the complexity of existing \mbox{$2$-SAT} algorithms: As bits have only a finite range of values, traversing a chain of implications in the implication graph poses no precision issues.
In the quantum setting, however, such a traversal involves computing products of $O(n)$ transfer matrices over some field extension of the rationals.
Trial assignments resulting from these products may require $O(n)$ bits per entry to represent; testing whether two possible assignments are equivalent may involve multiplying pairs of $n$-bit integers.
This is the source of the $M(n)$ term in the bit complexity estimate of Theorem~\ref{thm:main}.
To compare, similar considerations applied to Bravyi's \tq\ algorithm gives an upper bound of $O(n^4M(n))$ bit operations.

{It is not obvious that a faster runtime in terms of bit complexity should be possible in general.
As we show in Section~\ref{scn:lower bounds}, it is simple to construct a \tq\ instance with $m\in O(n)$ and whose unique product state solution requires $\Theta(n^2)$ bits to write down.
Thus, among algorithms which explicitly output the entire solution, our algorithm is optimal up to $\log$ factors, {taking time {$O(nM(n))\in \widetilde{O}(n^2)$} for $M(n)\in O(n \log(n) \,2^{O(\log^*(n))})$~\cite{F07}}.
Furthermore, as we also show in Section~\ref{scn:lower bounds}, for any algorithm $\mathbf A$ for \tq\ which produces the marginal of a satisfying solution (if one exists) on a single qubit in reduced terms\footnote{%
	N.B.\ Our algorithm \solveq\ is not such an algorithm, as the output may include cancellable factors in its representation.
}, there is a linear-time reduction from multiplication of $n$ bit integers to {the problem solved by $\mathbf A$}.
It follows that such an algorithm $\mathbf A$ must run in time $\Omega(M(n))$. {As discussed in} Section~\ref{scn:lower bounds}, this implies that unless $M(n)\in O(n)$, there is no general algorithm for \tq\ with linear bit complexity if the output is required to be in reduced form.
}


Theorem~\ref{thm:main} does reveal a special case in which we \emph{can} recover linear bit complexity: namely, when all constraints are product operators, {such as classical 2-SAT instances.} 
This special case still has essentially quantum features.
{For instance, the} phase transitions for satisfiability and for counting complexity in random product-constraint instances more closely resemble those of \tq\ than of classical \mbox{2-SAT}~\cite{Beaudrap14}.
Furthermore, ``YES'' instances of product-constraint \tq\ (e.g. such as instance $\{ \ket{00}\bra{00},\, \ket{11}\bra{11},\, \ket{++}\bra{++}\}$ on two qubits) may not be satisfiable by product states (see Section~\ref{scn:preliminaries} for how our algorithm deals with such instances).

{We close with open questions: is the bit-complexity of $O((n+m)M(n))$ for producing explicit assignments optimal?
Is there an $O(M(n))$ upper bound for producing representations of marginals of satisfying assignments?}

\paragraph{Organization of this paper.} In Section~\ref{scn:preliminaries}, we give notation, definitions, and the basic framework for our analysis (including transfer matrices). Section~\ref{scn:reductions} presents a series of lemmas and theorems to demonstrate how to overcome the obstacles presented in this introduction, and which form the basis of a proof of correctness for our algorithm \solveq. Section~\ref{scn:alg} states \solveq. Section~\ref{scn:runtime} sketches bounds on the runtime of \solveq~in terms of the field operations and bit operations. Additional technical details are deferred to Appendix~\ref{apx:runtime}. Section~\ref{scn:lower bounds} discusses lower bounds on the bit complexity of \tq.


\section{Preliminaries}\label{scn:preliminaries}

We begin by setting notation, stating definitions, and laying down the basic framework for our algorithm, including details on transfer matrices.

\paragraph{Notation.} The notation $:=$ denotes a definition and $[n]:=\set{1,\ldots,n}$. The vector space of (possibly non-normalised) single-qubit pure states is denoted $\cH := \C^2$. For a string $x=x_1x_2\cdots x_n\in\set{0,1}^n$, we write $\ket{x}:=\ket{x_1}\otimes\cdots\otimes\ket{x_n}$. For a vector space $\mathcal X$ over $\C$, we write $\lin{\mathcal X}$ 
for the set of linear 
operators on $\spa{X}$. The {nullspace of an} operator $A$ is denoted $\ker(A)$. For vectors $\ket{\psi}$ and $\ket{\phi}$, we write $\ket{\psi}\propto\ket{\phi}$ if $\ket{v}=c\ket{w}$ for \emph{non-zero} $c\in\complex$; if we wish to also allow $c=0$, we write ${\ket{\psi}\propto^\ast\ket{\phi}}$ instead. The latter two definitions extend straightforwardly to matrices. Given $\ket{\psi}\in\cH$, we write $\ket{\smash{\psi^\perp}}$ for the unique vector (up to scalar factors) {which is} orthogonal to $\ket{\psi}$.

{
\subsection{Quantum 2-SAT}
We now present a formal definition of quantum $k$-SAT (or $k$-QSAT).%
\begin{definition}[Quantum $k$-SAT~\cite{B06}]
    Let $n \ge k$ be an integer, and $\set{\Pi_i}_{i=1}^m\subset \lin{\cH^{\otimes k}}$ be a set of $k$-local orthogonal projection operators {(\ie,~of the form $I\otimes \bar{\Pi}_i$ for $k$-qubit projectors $\bar{\Pi}_i$)} with coefficients over some number field $\F$.
    \begin{description}[labelindent=2em,font=\mdseries\itshape]
    \item[Decision problem.] Does there exist a state $\ket{\psi}\in\cH^{\otimes n}$ such that $\Pi_{i}\ket{\psi}=0$ for all $i\in[m]$?
		\item[Search problem.] Produce a description of such a state $\ket{\psi}$ if it exists.
    \end{description}
\end{definition}
}
For precision reasons, we require in particular that the coefficients are drawn from a number field (a finite-degree field extension $\F = \Q[\omega]$).
We suppose that $\F$ is also specified as part of the input by means of a minimal polynomial $p \in \Q[x]$ for which $\F \cong \Q[x] / p$, together with a specification of how $\F$ embeds into $\C$~\cite{cohen93}.
(More details are given in Appendix~\ref{sscn:bitcomplexity}, where the runtime of the algorithm is carefully analyzed.)
{In the literature for \tq, one is usually more interested in how the structure of the placement of the projectors $\Pi_i$ affects the solution space, rather }than the complexity of the specification of $\F$ or the coefficients.
We therefore suppose that there is some constant $K$ which bounds from above the size of the specification of $\F$, and of the coefficients of the operators $\Pi_i$.

We next sketch how a $2$-SAT instance $\phi$ can be embedded into $2$-QSAT (cases $k > 2$ are similar). For each clause $C$ on boolean variables $(x_a,x_b)$, we define an operator $\Pi_C \in \lin{\cH^{\otimes 2}}$ of the form $\Pi_C := \ketbra{c_a}{c_a} \ox \ketbra{c_b}{c_b}$, where $c_a = 1$ if the variable $x_a$ is negated in $C$, and $c_a = 0$ otherwise; we fix $c_b$ similarly.
Then $\Pi_C$ is satisfied by $\ket{x_a x_b} \in \cH^{\otimes 2}$ if and only if $C$ is satisfied by $x_a x_b \in \{0,1\}^2$.
We extend $\Pi_C$ to an operator on $\cH^{\otimes n}$ by taking its tensor product with $I_2 \in \lin{\cH^{\otimes n-2}}$ on all tensor factors $i$ apart from $a,b \in [n]$.
Performing this for all clauses yields an instance of \tq, $\{ \Pi_C \}$, {in which all of the projectors are product operators (as mentioned in Section~\ref{scn:intro}),} and which imposes the same constraints on standard-basis vectors $\ket{x}$ as the clauses $C$ impose on $x \in \{0,1\}^n$.
Furthermore, as each $\Pi_C$ is positive semidefinite and diagonal, any $\ket{\psi}$ for which $\Pi_C \ket{\psi} = 0$ for all clauses $C$ must be a linear combination of vectors $\ket{x}$ which also satisfy $\Pi_C \ket{x} = 0$ for all $C$.
Thus this instance of \tq\ is satisfiable if and only if the original instance of \mbox{2-SAT} is, in which case there is a bijection between the solution space of the \mbox{2-SAT} instance and a basis for the solution-space of the \tq\ instance.

Finally, for a given \tq\ instance, we denote by $G$ its (potentially infinite) \emph{implication graph} (defined in Section~\ref{scn:intro}), and by $G'$ its \emph{interaction graph}, whose vertices are labelled by qubits, and with a distinct edge between vertices $i,j$ for each projector acting on them.

\paragraph{Reduction to cases satisfied by product states.}
We mainly consider product-state solutions to instances of \tq, in spite of instances (such as those described in ``Significance and open questions'' in Section~\ref{scn:intro}) in which no product state can be a solution.
A paradigmatic example is given by a single constraint $\Pi_\ast = I_4 - \ketbra{\Psi^-}{\Psi^-}$, where $\ket{\Psi^-} := {(\ket{01} - \ket{10})/\sqrt{2}}$; the unique satisfying assignment is the entangled state $\ket{\Psi^-}$.
Chen~\etal~\cite{CCDJZ11} nevertheless show that all instances of \tq\ are ``almost'' product-satisfiable in the following sense: The only pairs of qubits $(i,j)$ which are entangled for all satisfying states are those for which the sum of all constraints on $(i,j)$ is an operator $S_{ij}$ of rank $3$ (as with $\Pi_\ast$ above). We may treat such pairs by imposing the unique assignment $\ket{\psi_{ij}} \in \ker(S_{ij})$, and considering what restrictions this imposes on other qubits $k$ as a result of constraints on $(i,k)$ or $(j,k)$.
If we find no conflicts as a result of all such assignments, we obtain a sub-problem which is either unsatisfiable, or satisfiable by a product state.
(We describe this reduction in more detail in Section~\ref{scn:alg}.)




\paragraph{{Reduction to rank-1 instances.}} %
We may require that all constraints have rank $1$ (but possibly with multiple constraints on pairs of qubits%
),
by decomposing projectors $\Pi_{ij}$ of higher rank into rank-1 projectors $\Pi_{ij,1}$\,, $\Pi_{ij,2}$\,, \ldots, for which $\Pi_{ij} = \sum_k \Pi_{ij,k}$.
By the preceding reduction to product-satisfiable constraints, there will then be at most two independent constraints acting on any pair $(i,j)$.

\subsection{Transfer matrices}\label{sscn:reddefs}

A central tool in this work is the \emph{transfer matrix}, which for product states generalizes the equivalence between ${(x_i\vee x_j)}$ and ${(\bar x_i\Rightarrow x_j)}\wedge{(\bar x_j\Rightarrow x_i)}$ for bits.
Consider a rank-1 constraint $\Pi_{ij} = \ketbra{\phi}{\phi}$ on qubits $i$ and $j$, where $\ket{\phi}$ has Schmidt decomposition $\ket{\phi} = \alpha \ket{a_0}\ket{b_0} + \beta\ket{a_1}\ket{b_1}$. Then, the \emph{transfer} matrices $\xT_{\phi,ij},\xT_{\phi,ji}\in\lin{\complex^2}$ from $i$ to $j$ and from $j$ to $i$ are respectively given by:
{\begin{align}\label{eqn:transfer}
  \xT_{\phi,ij} &= \beta \ket{b_0}\bra{a_1} - \alpha \ket{b_1}\bra{a_0}, &
  \xT_{\phi,ji} &= \beta \ket{a_0}\bra{b_1} - \alpha \ket{a_1}\bra{b_0}.
\end{align}}
(When the state $\ket{\phi}$ is clear from context, we simply write $\xT_{ij}$ and $\xT_{ji}$.)
Given any assignment $\ket{\psi_i}\in\complex^2$ on qubit $i$, the transfer matrix $\xT_{\phi,ij}$ prescribes which single-qubit states $\ket{\psi_j}$ on $j$ are required to satisfy $\Pi_{ij}$, {via the constraint
$
 \ket{\psi_j} \propto^\ast \xT_{\phi,ij}\ket{\psi_i}.
$}
If $\xT_{\phi,ij}\ket{\psi_i}\neq 0$, then $\ket{\psi_j}$ is uniquely determined (up to equivalence by a scalar factor). This is guaranteed when $\ket{\phi}$ has Schmidt rank $2$, as $\xT_{\phi,ij}$ then has full rank. On the other hand, if $\xT_{\phi,ij}\ket{\psi_i}= 0$, then $\Pi_{ij}$ is satisfied for any assignment on $j$, so that $j$ remains unconstrained. This situation may only occur if $\ket{\phi}$ is a product constraint, so that $\xT_{\phi,ij}$ has a nullspace of dimension $1$.
This generalises the effect in the classical setting, that assigning ${x_i := 1}$ satisfies the constraint ${C = (x_i \vee x_j)}$, regardless of the value of $x_j$: the corresponding constraint and transfer matrix are $\ket{\phi} = \ket{00}$ and $\xT_{\phi,ij} = -\ketbra{1}{0}$, respectively.

\paragraph{Walk and cycle matrices.}
We take the closure of the transfer matrices, under composition along walks in the graph.
For any walk $W=(v_1,v_2,\ldots v_k)$ in a graph $G=(V,E)$, multiplying the transfer matrices $\xT_{v_{k\text{--}1}v_{k}}\cdots\xT_{v_2v_3}\xT_{v_1v_2}$ yields a new transfer matrix $\xT_{W}$, which we call the \emph{walk matrix}{ of $W$ (or \emph{path matrix}, if $W$ is a path).
For such a walk $W$, define $W\rev := (v_k,v_{k-1},\ldots,v_2,v_1)$. If a transfer matrix $\xT_W$ has singular value decomposition $\xT_W = s_0 \ketbra{\ell_0}{r_0} + s_1 \ketbra{\ell_1}{r_1}$, one may show by induction on the length of $W$ that
\begin{equation}
	\xT_{W\rev} = \pm \bigl(s_0 \ketbra{r_1}{\ell_1} + s_1 \ketbra{r_0}{\ell_0}\bigr),
\end{equation}
where the sign depends on whether $W$ has odd or even length.
In particular, this implies that
$ 
    \xT_W \xT_{W\rev} = \pm s_0 s_1 I.
$ 
Thus $\xT_W \xT_{W\rev} \propto^\ast I$ for all walks $W$, with a proportionality factor of zero if and only if $\xT_W$ is singular.
In particular, walk operators can sometimes be composed to represent ``cancellation'' of edges: For walks $U_1 = W' W$ and $U_2 = W\rev W''$, if $\xT_W$ is invertible, we have
$
	\xT_{W' W''} \,\propto\, \xT_{W''} \xT_{W\rev} \xT_W \xT_{W'} \,=\, \xT_{U_1 U_2},
$
representing a form of composition of walks in which repeated edges $(ij)(ji)$ cancel.

For $C=(v,u_1,u_2,\ldots,u_k,v)$ a cycle in $G$, the \emph{cycle matrix of $C$ at $v$} is just the walk operator $\xT_C$ arising from the walk from $v$ to itself along $C$.
We consider the cycles $C$ and (e.g.) $C'=(u_1,u_2,\ldots,u_k,v,u_1)$ to be distinct as walks; in particular, $C$ and $C'$ may give rise to distinct cycle matrices $\xT_{C'} \not\propto \xT_C$, which in any case represent operators on the state-spaces of distinct qubits.

Walk operators (and cycle operators in particular) allow us to more easily express express long-range constraints implicit in the original projectors $\Pi_{ij}$ (as one may show by induction):
\begin{lemma}[Inconsistency Lemma]
	\label{l:inconsistency}
	Let $W = (v, v_1, v_2, \ldots, v_\ell, w)$ be a walk in $G'$ with walk operator $\xT_W$, and let $\ket{\Psi} \in \cH^{\otimes n}$ be a product of single-qubit states $\ket{\psi_v}$ for each $v \in [n]$.
	If $\ket{\psi_w} \not\propto^\ast \xT_W \ket{\psi_v}$, then at least one constraint $\Pi_{ij}$ corresponding to an edge in $W$ is not satisfied by $\ket{\Psi}$.
\end{lemma}

\section{Efficient reductions via trial assignments in \tq}\label{scn:reductions}

As outlined in Section~\ref{scn:preliminaries}, we consider rank-1 instances of \tq\ which either have a product solution or are unsatisfiable.
In this section, we describe a means to incorporate transfer matrices into an efficient algorithm for \tq\ via the notion of a \emph{chain reaction}: An EIS-style subroutine for trial assignments.

As in Section~\ref{scn:intro}, we define the \emph{implication graph} of a $2$-QSAT instance to be an (infinite) directed graph $G=(V,E)$, where $V$ is the set of pairs $(i,\ket{\psi})$ for qubits $i$ and (distinct) states $\ket{\psi} \in \cH$.
There is a directed edge $(i,\ket{\psi}) \to (j,\ket{\phi})$ if and only if there is a constraint $\Pi_{ij}$ with transfer matrix $\xT_{ij}$ such that $\xT_{ij}\ket{\psi}\propto\ket{\phi}$. A ``chain reaction'' is a depth-first exploration of the nodes of $G$:

\begin{definition}[Chain reaction (CR)]
\label{def:CR}
For a qubit $i$ and state $\ket{\psi_i}\in\cH$, to \emph{induce a chain reaction (CR) at $i$ with $\ket{\psi_i}$} means to ``partially traverse'' $G$, starting from $(i,\ket{\psi_i})$ and keeping a record of the vertices $(u,\ket{\psi_u})$ seen for each $u$.
{This traversal is governed by a depth-first search (DFS) in the interaction graph $G'$, as follows.
For each vertex $(u,\ket{\psi_u})$ visited and each edge $\{u,v\}$ in $G'$, compute $\xT_{uv} \ket{\psi_u}$.
If this vector is non-zero, let $\ket{\psi_v} := \xT_{uv} \ket{\psi_u}$, and traverse to $(v,\ket{\psi_v})$ in $G$.
For any vertex $(v,\ket{\psi_v})$ visited by the CR, we say that the CR \emph{assigns $\ket{\psi_v}$ to $v$}.
In the sequence of vertices in $G$ visited by the CR, we may refer to instances of vertices $(v,\ket{\psi})$ for a given $v \in V$ as the \emph{first assignment}, the \emph{second assignment}, \etc\ made to $v$ by the CR.}
\end{definition}
\noindent Edges of $G'$ (and walks in $G'$) which are traversed by the depth-first search (DFS) governing a chain reaction, are also said to be traversed by the chain reaction (CR) itself.

The role of CRs in our analysis is to reveal constraints imposed by transfer matrices in an efficient manner.
Specifically, if the DFS in $G'$ which governs the CR encounters a cycle, it will visit a vertex $v$ in $G'$ twice, and so makes ``assignments'' to $v$ more than once.
If these assignments do not match, we say the CR has a \emph{conflict}. If {no such conflicts occur}, the CR is called \emph{conflict-free}.
(In either case, it does not continue the traversal of the CR from {the second, third, \etc\ assignments.})
{We formalise the intuitive significance of conflicts as follows:
\begin{lemma}[Conflict Lemma]\label{l:CRconflict}
	If a CR induced at $v$ with $\ket{\psi_v} \in \cH$ has a conflict, then no product state $\ket{\Psi} \in \cH^{\otimes n}$ for which the state of $v$ is $\ket{\psi_v}$ is a satisfying assignment.
\end{lemma}
\begin{proof}
  A conflict in the CR indicates the presence of two walks $W_1$ and $W_2$ in the interaction graph $G'$, from $v$ to some vertex $w$, for which $\xT_{W_1} \ket{\psi_v} \not\propto^\ast \xT_{W_2} \ket{\psi_v}$.
  It follows from the Inconsistency Lemma (Lemma~\ref{l:inconsistency}) that any product state in which $v$ takes the state $\ket{\psi_v}$ is not satisfying.
\end{proof}}
	
With the concept of CRs in hand, we can present the key ideas used by our algorithm.
First, conflict-free CRs yield partial assignments, which preserve the satisfiability of the instance defined on the remaining unassigned qubits.
Second, if a \tq\ instance is satisfiable, then a conflict-free CR can be found efficiently.
Our algorithm (presented in Figure~\ref{fig:algorithmq}) essentially operates by repeatedly finding conflict-free CRs, and removing the qubits given assignments by each CR, until either a conflict is detected (in which case we reject), or no unassigned qubits remain (in which case we accept).

\subsection{Using conflict-free chain reactions to remove qubits}\label{sscn:CR1}

The main result in this Section is Theorem~\ref{thm:setandforget} (Set-and-Forget Theorem), {which is essentially the converse of Lemma~\ref{l:CRconflict}, and allows us to reduce instances of \tq\ by providing partial solutions obtained from a CR induced on a single qubit.}

We begin by proving a correspondence between CRs and walk operators, in the sense that if there is a walk $W = (v, v_1, v_2, \ldots, w)$ in $G'$, and if $\ket{\psi_v} \notin \ker(\xT_W)$, a CR induced at $v$ with a state $\ket{\psi_v}$ should assign $\xT_W \ket{\psi_v}$ to $w$.
The obstacle here is that the CR might not traverse any of the edges of $W$ before assigning a state to $w${; we must then} relate $W$ to other walks in $G'$.
We do so as follows.
\begin{lemma}[Unique Assignment Lemma]\label{l:uniqueassignment}
    Suppose there exists a state $\ket{\psi}$ and a walk $W$ in $G'$ from $v$ to $w$ such that $\xT_{W}\ket{\psi}\propto\ket{\phi}$. Then, for any conflict-free CR induced on $v$ with $\ket{\psi}$, $w$ is assigned $\ket{\phi}$.
\end{lemma}
\begin{proof}	
	We show that there is a walk $\tilde W$ in $G'$ which is followed by the CR, for which $\xT_{\tilde W} \ket{\psi} \propto \ket{\phi}$.
  Suppose $W = (v, u_\ell, \ldots, u_1, u_0)$ for $u_0 := w$.
  For each $i \ge 0$, let $W_i$ denote the segment $(v, u_\ell, \ldots, u_i)$ of the walk $W$.
	Let $m$ be the smallest integer such that the CR traverses $W_m$.
	If $m = 0$, then we may take $\tilde W = W$ is the walk followed by the CR from $v$ to $w$.
	Otherwise, we show a reduction to ``deform'' $W$, to obtain walks $W'$, $W''$, \ldots, and a decreasing sequence $m > m' > m'' > \cdots$, for which the CR follows the walks $W_m$, $W'_{m'}$, $W''_{m''}$, \etc.
	These walks have successively shorter ``tails'' of edges which are not followed by the CR: the final such walk $\tilde W$ is then one which is completely followed by the CR.
	

	Given that $m > 0$, let $\ket{\psi_m} = \xT_{W_m} \ket{\psi}$.
	By hypothesis, the CR does not traverse the edge $(u_m, u_{m{-}1})$, either {because $\xT_{u_m u_{m\text{--}1}} \ket{\psi_m} = 0$, or because of} an assignment on $u_{m{-}1}$.
	The former implies $\xT_W \ket{\psi} = 0 \not\propto \ket{\phi}$, contrary to hypothesis.
	Then there is a walk $W'_{m{-}1} = (v, u'_r, \cdots, u'_m, u_{m{-}1})$ in $G'$, which is followed by the CR to make the assignment to $u_{m{-}1}$. {(Note that the assignments to $u_{m-1}$ made by both $W$ and $W'_{m-1}$ are proportional to one another, as otherwise the CR would have detected a conflict when attempting to traverse edge $(u_m,u_{m-1})$ during its breadth-first search.)}
	We extend the walk $W'_{m{-}1}$ to a walk $W' = (v, u'_r, \ldots, u'_m, u_{m{-}1}, \ldots, u_1, w)$.
	{The CR has traversed $W'$} at least as far as the vertex $u_{m{-}1}$, missing out fewer edges at the end than it does for $W$.
	Furthermore, as the CR is conflict-free, we have $\xT_{u_1 w}\xT_{u_2 u_1}\cdots \xT_{u_m u_{m\text{--}1}} \ket{\psi_m} \propto \xT_{W'} \ket{\psi}$, so that $\ket{\phi} \propto \xT_W \ket{\psi} \propto \xT_{W'} \ket{\psi}$ by construction.
	
	Repeating the reduction above yields a walk $\tilde W$ in $G'$ which is completely followed by the CR, for which $\xT_{\tilde W} \ket{\psi} \propto \ket{\phi}$ by induction.
	Then $\ket{\phi}$ is the assignment made to $w$ by the CR.
\end{proof}
\noindent
Note that the above result holds regardless of which walk $W$ we consider from $v$ to $w$, so long as $\xT_W \ket{\psi} \ne 0$.
Thus a conflict-free CR induced at $v$ depends on a consistency between all walk operators, from $v$ to any other given $w$, relative to the initial assignment $\ket{\psi_v}$.
For the case $w = v$, we then have:
\begin{lemma}[{Circuit Lemma}]\label{l:circuit}
    Let $W$ be a closed walk starting and ending at $v$.
    If $\ket{\psi_v}$ is not an eigenvector of $\xT_W$, then inducing a CR at $v$ with $\ket{\psi_v}$ yields a conflict.
\end{lemma}
\begin{proof}
	By definition, the CR assigns $\ket{\psi_v}$ to $v$.
  If the CR is conflict-free, then either $\xT_W \ket{\psi_v} = 0$ or $\xT_W \ket{\psi_v} \propto \ket{\psi_v}$, by Unique Assignment (Lemma~\ref{l:uniqueassignment}).
  Thus, if $\ket{\psi_v}$ is not an eigenvector of $\xT_W$, such a CR will have a conflict.
\end{proof}

\noindent
Lemma~\ref{l:uniqueassignment} also allows us to decouple the set of vertices given assignments by a CR, from the rest:

\begin{lemma}[Unilateral Lemma]\label{l:unilateral}
    For any state $\ket{\psi}$ and vertex $v$, suppose that a CR $C_1$ induced at $v$ with $\ket{\psi}$ is conflict-free. Let $A$ denote the set of vertices given an assignment by $C_1$, and $\ket{\psi_a}$ denote the assignment made by $C$ at a given $a \in A$.
    Then, for any constraint $\Pi_{ab}$ for $a \in A$ and $b \in V \setminus A$ and for any $\ket{\phi} \in \cH$, $\Pi_{ab} \bigl(\ket{\psi_a} \ox \ket{\phi}\bigr) = 0$.
\end{lemma}
\begin{proof}
		For $a \in A$, the CR $C_1$ must discover a walk {$W = (v, v_1, v_2, \ldots, v_\ell)$ for $v_\ell := a$, such that for any sub-walk $W_i = (v, v_1, \ldots, v_i)$ for $1 \le i \le \ell$, we have $\xT_{W_i} \ket{\psi} \ne 0$.}
		The assignment made to $a$ by $C_1$ is then $\ket{\psi_a} := \xT_{W} \ket{\psi}$ by construction.
		Conversely, as $b \notin A$, it follows by the Unique Assignment (Lemma~\ref{l:uniqueassignment}) that all walks $W_\ast$ in $G'$ from $v$ to $w$ satisfy $\xT_{W_\ast} \ket{\psi} = 0$: this holds in particular for the walk $W' = (v, v_1, \ldots, a, b)$.
		Then $\xT_{ab} \ket{\psi_a} = 0$, which is to say that $\Pi_{ab} \bigl( \ket{\psi_a} \ox \ket{\phi} \bigr) = 0$ for all $\ket{\phi}$.
\end{proof}

\noindent
The Unilateral Lemma allows us to treat conflict-free CRs as ``set-and-forget'' subroutines, in which we establish partial assignments on a set of qubits which we may remove from an instance $\mathscr P = \{\Pi_{ij}\}_{ij \in E}$ of \tq, obtaining a simpler, equivalent instance $\mathscr P' \subset \mathscr P$. Formally, we have the following.

\begin{theorem}[Set-and-Forget Theorem]\label{thm:setandforget}
    Let $\mathscr P = \{ \Pi_{ij} \}_{ij \in E}$ be an instance of \tq\ with interaction graph $G'=(V,E)$. Suppose that $C$ is a conflict-free CR induced at $v\in V$ with $\ket{\psi_v}\in\cH$, and let $A$ denote the set of vertices given assignments by $C$. Let $\mathscr P'$ be a \tq\ instance obtained from $\mathscr P$ by removing all constraints acting on $A$. Then $\mathscr P$ is satisfiable by product states if and only if $\mathscr P'$ is.
\end{theorem}
\begin{proof}
		For a given $a \in A$, let $\ket{\psi_a}$ denote the assignment made by $C$ to $a$.
		By construction, the states $\ket{\psi_a}$ jointly satisfy all constraints between vertices in $a$; and by the Unilateral Lemma (Lemma~\ref{l:unilateral}), the states $\ket{\psi_a}$ also unilaterally satisfy constraints between vertices in $A$ and vertices in $V \setminus A$.
		If $\mathscr P'$ is satisfiable by a state $\ket{\Phi} = \bigotimes_{v \in V\setminus A} \ket{\phi_a}$, then $\mathscr P$ is satisfiable by $\ket{\Psi} = \bigl[\bigotimes_{a \in A} \ket{\psi_a}\bigr] \ox \ket{\Phi}$.
		For the converse, suppose that $\mathscr P$ is satisfiable by some state $\ket{\Psi'} = \bigotimes_{v \in V} \ket{\psi'_v}$ (which may not agree with the assignments made by $C$).
		Define $\ket{\Psi} = \bigl[\bigotimes_{a \in A} \ket{\psi_a}\bigr] \ox \bigl[\bigotimes_{v \in V \setminus A} \ket{\psi'_v} \bigr]$.
		Again, $\ket{\Psi}$ satisfies all constraints acting on vertices $a \in A$, and by construction it also satisfies all constraints internal to $V \setminus A$.
		Then $\ket{\Psi}$ also satisfies $\mathscr P$, and its restriction to $V \setminus A$ satisfies $\mathscr P'$.
\end{proof}

\subsection{How to find conflict-free chain reactions efficiently}\label{sscn:discrete}

The Set-and-Forget Theorem (Theorem~\ref{thm:setandforget}) provides us with the following approach to find a product assignment for an instance $\mathscr P$ of \tq: \mbox{(\textit{i})}~pick an unassigned vertex $v$, \mbox{(\textit{ii})}~find $\ket{\psi_v}$ such that the CR induced at $v$ with $\ket{\psi_v}$ is conflict-free, and \mbox{(\textit{iii})}~use this CR to produce a partial assignment, reducing to an instance $\mathscr P'$ with fewer qubits.
It remains to attempt to find such a state $\ket{\psi_v}$, or determine that none exist, from the continuum $\cH$ of single-qubit states.

As we describe in Section~\ref{scn:intro}, and as shown by the {Circuit Lemma} (Lemma~\ref{l:circuit}), it suffices for us to restrict our search for $\ket{\psi}$ to the eigenvectors of $\xT_W$ for a closed walk $W$, \eg~a cycle.
Define a \emph{discretizing cycle} as a directed cycle $C$ (starting and ending at some vertex $v$) with cycle matrix ${\xT_C \not\propto^\ast I}$. {For such cycles, the Circuit Lemma} allows us to narrow down our search for $\ket{\psi_v}$ to the eigenvectors of $\xT_C$, of which there are at most two.
This raises two questions: (1) How to find discretizing cycles efficiently, and (2) how to deal with variables which are not on any discretizing cycle.

As noted in Section~\ref{scn:intro}, product operators complicate the task of detecting discretising cycles, but also provide a second way to narrow the search for assignments $\ket{\psi_v}$ leading to conflict-free CRs.
\begin{lemma}[Product Constraint Lemma]
	\label{l:prodconstr}
  In a product-satisfiable instance of \tq\ with a rank-1 product constraint projecting onto a state $\ket{\phi_{uv}} = \ket{\gamma_u} \ox \ket{\gamma_v}$, at least one of the CRs at vertex $u$ or $v$ with states $\ket{\smash{\gamma_u^\bot}}$ or $\ket{\smash{\gamma_v^\bot}}$, respectively, is conflict-free.
\end{lemma}
\begin{proof}
	Suppose that the instance is product satisfiable, but that a CR starting at qubit $u$ with state $\ket{\smash{\gamma_u^\bot}}$ has a conflict.
	Then by the Conflict Lemma (Lemma~\ref{l:CRconflict}), for any satisfying product state $\ket{\psi}=\bigotimes_{v \in V} \ket{\psi_v}$, we have $\ket{\psi_u} \not\propto \ket{\smash{\gamma_u^\bot}}$.
	By construction, we have $\ket{\psi_v} \propto \xT_{uv} \ket{\psi_u} = \ket{\smash{\gamma_v^\bot}} \ne 0$.
	Thus a CR induced at $v$ with $\ket{\smash{\gamma_v^\bot}}$ will be conflict-free {(as otherwise $\ket{\psi}$ cannot be a satisfying assignment, again by the Conflict Lemma)}.
\end{proof}
\noindent Using Lemma~\ref{l:prodconstr} together with the Set-And-Forget Theorem (Theorem~\ref{thm:setandforget}), we may find a partial assignment satisfying any given product constraint; repeating this for all product constraints will {either (\textit{i}) reveal that the original \tq\ instance is unsatisfiable, (\textit{ii}) yield a satisfying assignment for the entire instance, or (\textit{iii}) yield an equivalent instance of \tq\ in which all constraints are projectors onto \emph{entangled} states.}

Let us call an instance of \tq\ \emph{irreducible} if it has a connected interaction graph $G'$, and all of its constraints are rank-1 projectors onto entangled states.
In such an instance of \tq, all transfer matrices are invertible.
A conflict-free CR induced at any vertex will yield assignments for every other vertex; thus, a single discretizing cycle suffices to determine whether or not the instance is satisfiable.
We show that when a discretising cycle is present in such an instance of \tq, it is easily found:

\begin{lemma}
  \label{l:spanTreeDiscrete}
  Suppose $G'$ is an interaction graph of an irreducible instance of \tq, which contains a discretizing cycle $C$.
  Let $T \subset G'$ be a tree which contains all of the vertices of $C$.
  Then there is at least one edge $e$ in $C$, such that the (unique) cycle in the graph $T \cup \{e\}$ is a discretizing cycle.
\end{lemma}
\begin{proof}
 	In the tree $T$, there is a unique path $P_{vw}$ from any given vertex $v \in V$ to any other connected vertex $w$.
 	Furthermore, by the irreducibility of the \tq\ instance, $\xT_{P_{vw}}$ is non-singular in each case.
 	Suppose that $C = (v_1, v_2, \ldots, v_\ell, v_1)$ is a discretizing cycle in the implication graph $G$.
 	Consider the closed walk from $v_1$ to itself in $T$, given by
$	
		W = P_{v_1 v_2} P_{v_2 v_3} \cdots P_{v_\ell v_1}.
$
	By induction, we may show that the truncated walk $W' = P_{v_1 v_2} P_{v_2 v_3} \cdots P_{v_{\ell{-}1} v_\ell}$ satisfies $\xT_{W'} \propto \xT_{P_{v_1 v_\ell}} \!\propto \xT_{P_{v_\ell v_1}}^{-1}$ for each $\ell$: thus $\xT_W \propto I$.
	However,
$
		\xT_{C} = \xT_{v_\ell v_1} \cdots \xT_{v_2 v_3} \xT_{v_1 v_2} \not\propto I
$
	by hypothesis.
	Then there is an edge $vw$ in $C$ for which $\xT_{vw} \not\propto \xT_{P_{vw}}$.
	Then the unique cycle $C'$ in $T \cup \{vw\}$ contains the path $P_{vw}$ from $v$ to $w$, as well as the edge $vw$, and has cycle matrix $\xT_{C'} = \xT_{wv} \xT_{P_{vw}} \propto \xT_{vw}^{-1} \xT_{P_{vw}} \not\propto I$.
\end{proof}

\begin{theorem}[Cycle Discovery Theorem]
  \label{thm:DFSdiscret}
  Suppose $G'$ is the interaction graph of an irreducible instance of \tq, and contains a discretizing cycle $C$.
  Then a depth-first search from any vertex $v \in V$, in which each edge is traversed at most once, suffices to discover a discretizing cycle $C'$.
\end{theorem}

\begin{proof}
	Consider a DFS starting from any vertex $v \in V$.
	Define a tree $T \subset G$, in which each edge $e$ traversed by the DFS is included {if and only if $e$ is traversed for the first time some vertex is visited.}
	As the DFS reaches each vertex $w$, it also computes the path operator $\xT_{P_{vw}}$ for the path taken from $v$ to $w$.
	Each time the DFS traverses an edge $uw$ from some vertex $u$ to a vertex $w$ which it has previously visited, it tests whether $\xT_{\!\!\;P_{vu}} \propto \xT_{\!\!\;wu} \xT_{\!\!\;P_{vw}}$.
	If so, it continues the DFS from $w$.
	Otherwise the cycle $C'$ consisting of $P_{vu}\rev P_{vw}$ 
	followed by $wu$ is discretizing, as ${\xT_{C'} \propto \xT_{uw} \xT_{P_{vu}} \xT_{P_{vw}}^{-1} \not\propto^\ast I}$.
	Conversely by Lemma~\ref{l:spanTreeDiscrete}, if $G$ has a discretizing cycle, the DFS must eventually traverse such an edge.
\end{proof}

\noindent Implicit in Theorem~\ref{thm:DFSdiscret} is a linear-time algorithm for finding discretising cycles in an irreducible instance of \tq, when one is present.
It remains to describe how to treat irreducible instances which have no discretizing cycles.
The absence of any means of discretising the state-space of any qubit in such an instance actually represents freedom of choice in this case; while this is implicit in Refs.~\cite{B06,LMSS10,dBOE10}, we prove it here for the sake of completeness.
\begin{lemma}[Free Choice Lemma]
	\label{l:freechoice}
	In an irreducible instance of \tq\ with no discretizing cycles, any choice of single-qubit state $\ket{\psi_v}$ for some $v$ in the component gives rise to a conflict-free CR.
\end{lemma}
\begin{proof}
	Let $G'$ be the interaction graph.
	Consider a CR induced at $v$ with $\ket{\psi_v}$, and consider the paths $P_{vw}$ to each vertex $w$, by which the CR makes its first assignment $\ket{\psi_w} := \xT_{P_{vw}} \ket{\psi_v}$ to $w$.
	If $P'_{vw}$ is another walk from $v$ to $w$, we have $\xT_{P_{vw}\rev} \xT_{P'_{vw}} \propto I$, from the hypothesis that there are no discretising cycles; then $\xT_{P'_{vw}} \propto \xT_{P_{vw}}$.
	Thus, regardless of the choice of $\ket{\psi_v}$, a consistent assignment $\xT_{P'_{vw}} \ket{\psi_v}$ is computed every time the CR traverses an edge to visit $w$.
\end{proof}

\section{A linear-time \tq\ algorithm}\label{scn:alg}

We finally present our \tq\ algorithm in Figure~\ref{fig:algorithmq}, whose correctness follows immediately by combining the results of Section~\ref{scn:reductions}.
Following~\cite{EIS76}, we implement CRs (corresponding to their trial assignments) in parallel to ensure a linear bound on run-time; this is expanded upon in Section~\ref{scn:runtime}.

\begin{figure}[t]
		\begin{framed}
		\leftskip-0.25em
		\textbf{Input:}~
		\begin{minipage}[t]{0.85\textwidth}\raggedright
			An instance of \tq\ consisting of rank-$1$ projectors $\mathscr P = \set{\Pi_{ij}}$ with interaction graph $G'=(V,E)$, with at most two parallel edges $(u,v)$ per distinct $\{u,v\} \subset V$.
		\end{minipage}
		\smallskip
		\begin{enumerate}
		\item
			\emph{Discretize on product constraints} ---
			While there exists a projector $\Pi_{ij} = \ket{\phi_{ij}}\bra{\phi_{ij}}$ such that $\ket{\phi_{ij}} = \ket{\gamma_i} \ox \ket{\gamma_j}$ is a product state: simulate CRs at each $v \in \{i,j\}$ with $\ket{\gamma_v^\perp}$, in parallel.
			\begin{enumerate}
			\item
				If conflicts arise in both CRs, halt and {output ``\textsc{unsat}''.}
			\item
				Fix the assignments for the first conflict-free CR that terminates, remove the set $A$ of vertices that it visited from $G'$, and go to Step~1.
			\end{enumerate}

		\item
			\emph{Discretize on cycles} ---
				While there exists $v \in V$: search for a discretizing cycle $C \subseteq G'$ in the same connected component of $v$.
				\begin{itemize}[leftmargin=1.2em]
				\item
					If such a cycle $C$ is found at a vertex $u$: Let $\xT_C$ be its cycle matrix, and $S$ denote the set of eigenvectors of $\xT_C$.
					Simulate CRs at $u$ with each $\ket{\psi_u} \in S$, in parallel.
					\begin{enumerate}
					\item
						If conflicts arise in both CRs, halt and {output ``\textsc{unsat}''.}
					\item
						Fix the assignments for the first conflict-free CR that terminates, remove the set $A$ of vertices that it visited from $G'$, and go to Step~2.
					\end{enumerate}
			\item
				If no such cycle is found: Induce a CR at $v$ with $\ket{\psi_v} := \ket{0}$.
				Fix assignments made by the CR, remove the set $A$ of vertices that it visits from $G'$, and go to Step~2.
			\end{itemize}
		\item
			\emph{Normalize} --- For each qubit $v$, compute whether the assignment $\ket{\psi_v}$ is normalised: if not, compute a normalised version $\ket{\psi_v} := \ket{\psi_v} \big/ \sqrt{\bracket{\psi_v}{\psi_v}}$.
		\end{enumerate}
		\textbf{Output:}~
		\begin{minipage}[t]{0.85\textwidth}\raggedright
			{``\textsc{unsat}'',} or unit vectors $\ket{\psi_v} \in \cH$ for each $v \in V$ which jointly satisfy $\mathscr P$.
		\end{minipage}
		\hspace{-8mm}
		\end{framed}
		\vspace*{-2ex}
		\caption{%
			\label{fig:algorithmq}
			An algorithm for \tq, denoted \solveq.
		}
\end{figure}

\paragraph{Preprocessing stage to impose input constraints.} 
%
For conciseness, we present \solveq\ in Figure~\ref{fig:algorithmq} with restrictions on the inputs it takes.
As we indicate in Section~\ref{scn:preliminaries}, following Chen~\etal~\cite{CCDJZ11}, these restrictions ensure that the instance presented to \solveq\ is either satisfiable by a product state or unsatisfiable.
{These restrictions can can be imposed through a pre-processing phase, as follows.}
For each pair $\{u,v\}$ subject to multiple constraints, sum the projectors to obtain positive semidefinite operator $S_{uv}$.
Then perform the following:
\vspace{-0.5ex}
\begin{enumerate}[leftmargin=3ex, itemsep=0ex]
	\item
		If any pair $\{u,v\}$ has $\operatorname{rank}(S_{uv})=4$, halt with output \textsc{unsat} (as $\ker(S_{uv})$ contains no states).
	\item
		For each pair $\{u,v\}$ with $\operatorname{rank}(S_{uv})=2$, replace the constraints on $\{u,v\}$ with $\Pi_{uv,1} = \ketbra{\eta_1}{\eta_1}$ and $\Pi_{uv,2} = \ketbra{\eta_2}{\eta_2}$, 
		for linearly independent columns $\ket{\eta_1},\ket{\eta_2}$ of $S_{uv}$.
	\item
		For each pair $\{u,v\}$ with $\operatorname{rank}(S_{uv})=3$, record the unique state $\ket{\psi_{uv}}$ which spans $\ker(S_{uv})$ as a joint assignment to $(u,v)$, and remove the constraints on $\{u,v\}$.		
		If $\ket{\psi_{uv}} = \ket{\psi_u} \ox \ket{\psi_v}$, record $\ket{\psi_u}$ and $\ket{\psi_v}$ as assignments to $u$ and $v$ respectively.
		(If any qubit is subject to conflicting assignments, halt with output \textsc{unsat}.)
	\item
		For each pair $\{u,v\}$ given an assignment $\ket{\psi_{uv}}$ in the preceding step:
\vspace{0.5ex}		
\begin{compactitem}
		\item
			If $\ket{\psi_{uv}} = \ket{\psi_u} \ox \ket{\psi_v}$, induce CRs (sequentially) at $u$ with $\ket{\psi_u}$ and at $v$ with $\ket{\psi_v}$.
\vspace{0.5ex}		
		\item
			If not, and there are non-product constraints $\Pi_{iu}$ or $\Pi_{iv}$ for any $i$, halt with output \textsc{unsat}
			(as any state of $i$ is compatible only with product states on $\{u,v\}$).
			Otherwise, for each $\Pi_{iu} = {\ketbra{\gamma_i}{\gamma_i} \ox \ketbra{\gamma_u}{\gamma_u}}$ or $\Pi_{iv} = {\ketbra{\gamma_i}{\gamma_i} \ox \ketbra{\gamma_v}{\gamma_v}}$, induce a CR (sequentially) at $i$ with $\ket{\smash{\gamma_i^\bot}}$.
		\end{compactitem}
		For any CR induced, halt (with output \textsc{unsat}) either if the CR has a conflict, or if it makes an assignment to some other qubit $w$ which has been given a different assignment as a result of a rank-3 constraint.
		If no conflict is detected, record the assignments, and remove the set $A$ of qubits given assignments from $G'$.
	\end{enumerate}
This preprocessing phase involves much the same subroutines as \solveq\ itself, and does not contribute to the asymptotic run-time.
(We include these steps in our detailed runtime analysis in Appendix~\ref{apx:runtime}.)
	

\section{Runtime analysis}\label{scn:runtime}
We briefly sketch the runtime analyses for \solveq\ in terms of field operations over $\complex$ and bit operations, and discuss an optimization for the setting of product state constraints. A more in-depth treatment is given in Appendix~\ref{apx:runtime}. We assume a random-access machine, so that memory access {takes} unit time. 
{The constraints $\Pi_{i}$ are specified as $4 \x 4$ matrices with coefficients from a finite-degree field extension $\F{\!\;:\!\;}\Q$, whose specification is also part of the input; arithmetic operations over such number fields can be performed efficiently~\cite{cohen93}.
(Further details given in Appendix~\ref{sscn:bitcomplexity}.) From this representation we extract the basis vectors $\ket{\eta_i}$ for the image of $\Pi_{i}$ by taking columns of $\Pi_{i}$, and omit normalisation: \solveq\ then uses $\ket{\eta_i}$ to represent $\Pi_i$.
Vectors are only normalised as the final step of the algorithm.}

\paragraph{Field operations.} \solveq\ requires $O(n+m)$ operations over $\complex$, for $n$ and $m$ the number of variables and clauses, respectively. {As each vector $\ket{\eta_i}$ is in $\C^4$, operations on them (such as determining if $\ket{\eta_i}$ is a product constraint in Step~1) require $O(1)$ field operations. Following EIS~\cite{EIS76}, we simulate CRs in parallel by interleaving their steps, terminating both simulations as soon as one of them is found to be conflict-free. In the preprocessing phase and in Step 1b, this ensures that the number of vertices and edges removed (upon completion of a conflict-free CR) is proportional to the number of vertices and edges visited during the parallel CRs. Hence, the total number of edge-traversals of \solveq\ is $O(m)$. Finally, by Step 2, the instance has been simplified to a disjoint union of irreducible instances. Theorem~\ref{thm:DFSdiscret} ensures that if a discretizing cycle exists in any of the components, it can be found by a depth-first search; moreover, a single conflict-free CR suffices to assign satisfying states to all vertices in each component.}

\paragraph{Bit operations.} The bit-complexity of \solveq\ differs from the field-operation complexity, for the simple reason that multiplying $k$ transfer matrices yields a path matrix with $O(k)$-bit entries.
Thus, {operations such as determining the eigenvectors of such matrices, or whether $\ket{\psi} \propto \ket{\phi}$ for vectors in the image of these matrices, can take time $O(M(k))$, where $M(k)$ is the time to multiply two $k$-bit integers.
This follows from the fact that computing $\sqrt D \in \Z$ for a perfect square $D \in \Z$ can be performed in $O(M(n))$ time using Newton's method {(see \eg~Theorem~9.28 of Ref.~\cite{GG03})}; and that equality testing over $\Q$ is bounded by $O(M(n))$, for rationals $r,s \in \Q$ with $n$ bit representations as ratios. (To test whether $\tfrac{a}{b}$ and $\tfrac{c}{d}$ are equal, one tests whether $ad - bc = 0$.)
Since the number of times we might need to compute eigenvectors or decide proportionality may scale as $m+n$, the runtime of $O((m+n)M(n))$ follows.}

It may be necessary for \solveq\ to represent its output using further field extensions $\E {\;\!:\;\!} \F$, for instance, when solving the characteristic polynomial $\det(\lambda I - \xT_C)$ of a cycle matrix $\xT_C$, if the discriminant $D = (\Tr \xT_C)^2 + 4(\det \xT_C)$ is not a perfect square in $\F$.
(We discuss this aspect of the algorithm in Appendix~\ref{sscn:bitcomplexity}.)
However, by the Set and Forget Theorem~\ref{thm:setandforget}, any extension required by a CR will be independent of the CRs involved in the assignments made by other CRs; furthermore, the extensions involved in each CR is only quadratic, and specifically by a square root $\sqrt D$ of an element $D \in \F$.

The approach taken to the quadratic extensions by \solveq\ is unconventional.
Specifically, {unless $D\in\Q$,} we do not evaluate whether or not $\sqrt D$ is in $\F$ before defining the (possibly trivial) ``extension'' $\E = \F[\sqrt D]$.
That is, we allow representations of number fields $\F_k := \Q[\omega_1, \omega_2, \ldots, \omega_k]$ in which $\omega_j = \sqrt{\alpha_j}$ for some $\alpha_j \in \F_{j{-}1}$ (possibly including the case $\omega_1 = \sqrt{s}$ for $s \in \Q$), and where it may come to pass that $\omega_j \in \F_{j{-}1}$.
This prevents us from easily presenting coefficients in a normal form: crucially however, it is still possible for us to perform equality tests and arithmetic operations in time $O(M(n))$, for $\alpha \in \F_k$ expressed as $\tfrac{1}{\mu} f(\omega_1, \ldots, \omega_k)$ for $\mu \in \Z$ and $f \in \Z[x]$ with coefficients of size $O(n)$, provided that $k$ is bounded by a constant.
(In the case of \solveq, we bound $k \le 3$.)

Thus while the output of \solveq\ may not be reduced, it nevertheless presents exact, normalised, satisfying  states by means of tensor factors.
Complete details are to be found in Appendix~\ref{sscn:bitcomplexity}.

\paragraph{Reduced complexity of \tq\ for product constraints.} 
Using a simple optimization which exploits product constraints, \solveq\ can in fact accept inputs over any field extension ${\F\:\!{:}\:\!\Q}$ (algebraic or otherwise), and solve them with $O(n+m)$ bit operations provided that all projectors are product operators.
This requires only that arithmetic operations and equality testing against $0$ can be performed in $\F$ in $O(1)$ time on inputs with representations of size $O(1)$.
Specifically: the transfer matrix of a product constraint $\Pi_{uv} = \ketbra{\gamma_u}{\gamma_u} \ox \ketbra{\gamma_v}{\gamma_v}$ is $\xT_{uv} \propto \ket{\smash{\gamma_v^\bot}}\bra{\gamma_u}$, whose image is spanned by $\ket{\smash{\gamma_v^\bot}}$.
For any assignment $\ket{\psi_u}$ to $u$, if $\xT_{u,v} \ket{\psi_u} \ne 0$, we can set $v$ to $\ket{\smash{\gamma_v^\bot}}$ (which by assumption on the input requires $O(1)$ bits), as opposed to the potentially more complex vector $\xT_{u,v} \ket{\psi_u} \propto \ket{\smash{\gamma_v^\bot}}$.
Thus, in Step~1, the complexity of the assignments made by a CR are no more complex than the vectors of the projectors $\Pi_{uv}$ in the input, so that all algebraic operations may be performed in $\Theta(1)$ time rather than $O(M(n))$ time.
In particular, for classical \mbox{$2$-SAT} instances, we recover an $O(m+n)$ upper bound on the bit-complexity of \solveq, matching the asymptotic performance of the APT and EIS algorithms~\cite{APT79,EIS76}.

\section{On lower bounds for bit complexity}
\label{scn:lower bounds}

Most investigations into \tq\ are presented in terms of unit-cost operations over some algebraic number field $\F$.
As a result, no restrictions are usually put on how the output of a classical solution to \tq\ is represented.
To consider lower bounds on the bit-complexity of presenting a solution to \tq, it becomes necessary to consider what restrictions to impose on the output, as without such restrictions the notion of what form the output may take becomes ill-defined.
We impose the restriction of outputs which are \emph{rationalised}, as follows.
Let $\F = \Q[\omega]$ be algebraic number field, so that $\omega$ is an algebraic number whose minimal polynomial $p$ is a monic polynomial over $\Z$.
An element $\alpha \in \F$ is presented in \emph{rationalised form} by an expression of the form $f(\omega)/D = \alpha$, where $D > 0$ is an integer and $f \in \Z[x]$ is an polynomial such that $\deg(f) < \deg(p)$.
Despite the unconventional representation described in Section~\ref{scn:runtime}, this is one constraint which the output of \solveq\ respects.

There are further restrictions which one might consider, such as the output state vectors being normalised (which \solveq\ satisfies), and that they be \emph{reduced}: that the coefficients $\alpha = f(\omega)/D$ satisfy $\gcd(f,D) = 1$.
Consider, for instance, an algorithm which produces its output in \emph{minimal form}: each state that it outputs is normalised, in reduced rationalised form, and involves the minimal field extension ${\F\:\!{:}\:\!\Q}$ necessary to do so, represented as $\F = \Q[\omega]$ where the minimal polynomial of $\omega$ is a monic polynomial over the integers.
While \solveq\ does not compute outputs in minimal form (\eg,~it may fail to produce outputs in reduced form), we show that the multiplication time $O(M(n))$ for $n$ bit integers is a relevant lower bound for algorithms which do, suggesting that the role of $M(n)$ in the upper bound of \solveq\ is not merely accidental.
\begin{lemma}\label{l:last}
    There exist instances of \tq\ on $n$ vertices and $m\in O(n)$ clauses, such that exhibiting a requested tensor factor of a satisfying solution, in minimal form, requires $\Omega(M(n))$ bit operations in in the worst case.
\end{lemma}
\begin{proof}
    Let $M$ and $N$ be positive, odd $n$-bit integers, with binary expansions $M = \sum_{t=0}^{n-1} 2^t M_t$ and $N = \sum_{t=0}^{n-1} 2^t N_t$, where $M_i, N_i \in \{0,1\}$ for each $0 \le i < n$.
		We construct an instance of \tq\ {whose unique product state solution is one} in which one of the qubits $q$ is assigned a state
    \begin{equation}
    \label{eqn:lowerBoundAssignment}
			\ket{\psi_q} \,:=\, \frac{M}{\sqrt D} \ket{0} + \frac{2^n + MN}{\sqrt D}\ket{1} \,=\, \frac{M \sqrt D}{D} \ket{0} + \frac{(2^n + MN)\sqrt D}{D}\ket{1},
    \end{equation}
    where $D = M^2 + (2^n \!+\! MN)^2$.
    Either the middle or the right-hand expression in Eqn.~\eqref{eqn:lowerBoundAssignment} is in rationalised and normalised form, depending on whether $D$ is a perfect square.
    As $M$, $2^n + MN$, and $D$ are coprime, that rationalised expression is in reduced form, if $\F = \Q[\sqrt D]$.
    If $D$ is neither a perfect square nor square-free, it may be that $\sqrt D$ is represented as $\delta \sqrt{D'} \in \Q[\sqrt{D'}]$, where $D = D' \delta^2$.
    In this case, by hypothesis, a representation of $\ket{\psi_q}$ in reduced form would be identical (up to signs) to
    \begin{equation}
      \ket{\psi_q} = \frac{M\sqrt{D'}}{D' \delta} \ket{0} + \frac{(2^n + MN) \sqrt{D'}}{D' \delta} \ket{1}.
    \end{equation}
		In any case, the minimal form representation would provide a specification of the extension element $\sqrt{D'}$, the denominators $D'\delta$, and the numerators $A = M$ and $B = 2^n + MN$ (or $A = -M$ and $B = -2^n - MN$, which yields an equivalent vector in $\Q[\sqrt{D'}]$).
		From such a representation, one could compute $MN$ simply as $B - 2^n$ (or $-B -2^n$ respectively), which requires time $O(n)$.

		The instance we construct is on a chain of $2n+2$ qubits, labelled $v \in \{0,1,2,\ldots,2n{+}1\}$, as follows.
		For $1 \le i \le n$, we define matrices
		\begin{align}
				\xT_{i{-}1,i}
		  &=
				\begin{pmatrix}
		      1				&		0		\;\;	\\
		      M_{n-i} & 	2		\;\;
		    \end{pmatrix},
		  &
				\xT_{n{+}i,n{+}1{+}i}
			&=
				\begin{pmatrix}
		      1				&		0		\;\;			\\
		      N_{n-i} & 	2		\;\;
		    \end{pmatrix};
		\end{align}
		and also two matrices $\xT_{n,n{+}1}$ and $\xT'_{0,1}$: 
		\begin{align}
		  \xT_{n,n{+}1} &= \begin{pmatrix} 0 & 1 \\ 1 & 0 \end{pmatrix},
		&
			\xT'_{0,1} &= \begin{pmatrix} 0 & 1\\ 0 & M_{n{-}1} \end{pmatrix}.
		\end{align}
		For each $i \in \{0,1,2,\ldots,2n\}$, we include a constraint $\Pi_{i,i{+}1}$ between qubits $i$ and $i+1$, with transfer matrix $\xT_{i,i{+}1}$; and we also include a second constraint $\Pi'_{0,1}$ between $0$ and $1$, with transfer matrix $\xT'_{0,1}$.
		The resulting instance of \tq\ has two rank-1 constraints between qubits $0$ and $1$, and one rank-1 constraint between all other consecutive pairs of qubits.
		By Chen~\etal~\cite{CCDJZ11}, this instance is then satisfiable by a product state if it is satisfiable at all.
		It is easy to show that all of the projectors have rational coefficients in this case, so we take the field of the representation to be $\Q$ itself.

		We show that there is a unique product state which satisfies the above instance of \tq.
		It is easy to show that the opposite transfer operator to $\xT'_{0,1}$ is
		\begin{equation}
		  \xT'_{1,0} \propto \begin{pmatrix} -M_{n{-}1} & 1 \; \\ 0 & 0 \; \end{pmatrix}
		\end{equation}
		so that $\xT'_{1,0} \xT_{0,1} \propto \ketbra{0}{1}$.
		The only eigenvector of this operator is $\ket{0}$, which is therefore the only single-qubit state on qubit $0$ which is consistent with a satisfying solution.
    As all other transfer operators are non-singular, this determines a unique assignment for all other qubits $i$ in the chain, determined by the first column of the walk operator $\xT_{[0,i]} := \xT_{i{-}1,i} \cdots \xT_{1,2} \xT_{0,1}$.
    It is easy to show for $1 \le i \le n$ that
    \begin{equation}
      \xT_{[0,i]} = \begin{pmatrix}
              1 & \;\;0\;\;	\\[1ex]
							\;\;\sum\limits_{\mathclap{1\le t \le i}} M_{n{-}t} 2^{i-t} & \;\;2^i\;\;
            \end{pmatrix},
    \end{equation}
		and that in particular
		\begin{equation}
      \xT_{[0,n]} = \begin{pmatrix}
              1 & 0	\\
							M & 2^n
            \end{pmatrix};	
		\end{equation}
		from this we easily obtain
		\begin{equation}
      \xT_{[0,n+1]} = \begin{pmatrix}
              M & 2^n\\
							1 & 0
            \end{pmatrix};	
		\end{equation}
		from which point we may prove by induction for $1 \le i \le n$ that
		\begin{equation}
      \xT_{[0,n+1+i]}
      = \begin{pmatrix}
              M \;\;&\;\; 2^n\\[1ex]
							2^i + M \sum\limits_{\mathclap{1\le t \le i}} N_{n-t} 2^{i-t} \;\;&\;\; 2^n \sum\limits_{\mathclap{1\le t \le i}} N_{n-t} 2^{i-t}\;
            \end{pmatrix};	
		\end{equation}
		so that
		\begin{equation}
      \xT_{[0,2n+1]}
      = \begin{pmatrix}
              M \;&\; 2^n\\
							2^n + M N \;&\; 2^n N
            \end{pmatrix}.
		\end{equation}
		Let $q$ be qubit $2n+1$.
		The only assignment to this qubit which is consistent with a satisfying assignment is then the state given by the first column of $\xT_{[0,2n+1]}$, which is $M\ket{0} + {(2^n + MN)\ket{1}}$; the vector given by Eqn.~\eqref{eqn:lowerBoundAssignment} is the normalised version of this vector.
		
		Using the techniques of Laumann~\etal~\cite{LMSS10}, we may show that the space of satisfying assignments of this instance has dimension $2$, spanned by the product solution above, and an entangled solution on all of the qubits.
		Considering all projectors except for $\Pi'_{0,1}$, there is an invertible (non-unitary) local transformation mapping all projectors $\Pi_{i{-}1,i}$ to $\ketbra{\Psi^-}{\Psi^-}$, the two-qubit antisymmetric projector.
		Thus the satisfying states for these projectors are the symmetric subspace on $S = 2n+2$ qubits, which is spanned by any collection of states of the form $\ket{\alpha_i}^{\otimes 2n+2}$, for $S+1 = 2n+3$ distinct states $\ket{\alpha_i}$.
		Any state in this space which is not a product state, is entangled across the entire chain of qubits.
		Undoing this change of local co-ordinates, it follows that any state which satisfies the above instance of \tq\ which is not a product state, is also entangled across the entire chain of qubits {(\ie,~it cannot be factorized across any cut). Since we require each factor to be explicitly written in the standard basis, such a solution would then require explicitly writing out the standard basis elements of a vector of dimension $2^{2n+2}$}; such solutions would require vectors of dimension $2^{2n+2}$ to represent.
		Any algorithm which in polynomial time exhibits one of the tensor factors of the solution, must therefore exhibit factors of the product solution.
		In particular, it must compute $\ket{\psi_q}$ if this is the required tensor factor.
		As we have already shown an $O(n)$ reduction from computing the product $MN$ to computing the minimal representation of $\ket{\psi_q}$, it follows that there is an $\Omega(M(n))$ lower bound for such an algorithm in the worst case.
\end{proof}
\begin{corollary}
    If there does not exist a $\Theta(n)$-time algorithm for multiplying two $n$-bit integers, then there does not exist an $O(m+n)$-time algorithm to present single-qubit marginals of satisfying solutions to instances of \tq.
\end{corollary}

We would also like to show lower bounds for algorithms such as \solveq, which do not necessarily compute its output in reduced form, but which does compute an \emph{explicit} output, in the sense of presenting a complete description of a satisfying solution via tensor factors.
We may obtain such lower bounds even for algorithms which produce non-normalised outputs, as follows.

\begin{lemma}\label{l:spaceLowerBound}
    There exist instances of \tq\ on $n$ vertices and $m\in O(n)$ clauses, such that an explicit rationalised (but not necessarily normalised) assignment for a satisfying state requires $\Omega(n^2)$ bits.
\end{lemma}
\begin{proof}
		We may simplify the construction of Lemma~\ref{l:last} by omitting the qubits $n+1$, \ldots, $2n+1$ and the projectors which act on them.
		This yields an instance in which there is a unique product solution (with all other solutions requiring a vector of dimension $2^{n+1}$ to represent).
		In this product state, the qubit $n$ is in a state $\ket{\psi_n} \propto \ket{0} + M\ket{1}$.
		More generally, each qubit $1 \le i \le n$ is in a state
		\begin{equation}
		  \ket{\psi_i} \propto \ket{0} + M^{(i)} \ket{1}
		\end{equation}
		where $M^{(i)} = \sum_{t=1}^i M_{n-t}2^{i-t}$.
		As $M_{n{-}1} M_{n{-}2} \cdots M_2 M_1 \in \{0,1\}^{n-1}$ may be an arbitrary $n-1$ bit string, and as we require the tensor factors on the qubits $i$ to be presented independently of one another, the integers $M^{(i)}$ cannot be represented any more succinctly in the worst case; at best, by applying arbitrary scalar factors, we may consider representations $\ket{\psi_i} = \tfrac{1}{\alpha_i} \ket{0} + \tfrac{M^{(i)}}{\alpha_i} \ket{1}$, in which the representation of the $\ket{1}$ component of $\ket{\psi_i}$ may be reduced if $\alpha_i$ divides $M^{(i)}$, but at the cost of increasing the size of the representation of the $\ket{0}$ component. {(More formally, if the pair $(1/\alpha_i,M^{(i)}/\alpha_i)$ has asymptotically smaller Kolmogorov complexity than the pair $(1,M^{(i)})$, we would have a contradiction, since the former allows us to extract $M^{(i)}$ --- thus, we would have a shorter description of $M^{(i)}$ than its Kolmogorov complexity allows.)}
		Thus, for any constant $0 < \alpha < 1$, the qubits $\lfloor \alpha n \rfloor < i < n$ all require $\Omega(n)$ bits to represent, yielding a total lower bound of $\Omega(n^2)$.
\end{proof}
\begin{corollary}
  Up to $\Omega(\log(n)^{1 + o(1)})$ factors, \solveq\ is optimal among algorithms which present explicit expressions for satisfying assignments.
\end{corollary}

\section*{Acknowledgements}
This project was initiated while NdB was affiliated with the Centrum Wiskunde \& Informatica.
NdB thanks Rick, Linda, Colin, and Michelle de~Beaudrap for hospitality during an academic absence, and acknowledges support from the UK Quantum Technology Hub project NQIT.
SG thanks Ronald de Wolf and Centrum Wiskunde \& Informatica for their hospitality. SG is supported by NSF grant CCF-1526189.

\bibliographystyle{alpha}
\bibliography{Sevag_Gharibian_Central_Bibliography_Abbrv,Sevag_Gharibian_Central_Bibliography}

\newcommand{\etalchar}[1]{$^{#1}$}
\begin{thebibliography}{CCD{\etalchar{+}}11}

\bibitem[APT79]{APT79}
B.~Aspvall, M.~F. Plass, and R.~E. Tarjan.
\newblock A linear-time algorithm for testing the truth of certain quantified
  boolean formulas.
\newblock {\em Information Processing Letters}, 8(3):121--123, 1979.

\bibitem[ASSZ15]{ASSZ-2015}
I.~Arad, M.~Santha, A.~Sundaram, and S.~Zhang.
\newblock Linear time algorithm for quantum {2SAT}.
\newblock arXiv:1508.06340, 2015.

\bibitem[Bra06]{B06}
S.~Bravyi.
\newblock Efficient algorithm for a quantum analogue of 2-{SAT}.
\newblock Available at arXiv.org e-Print quant-ph/0602108v1, 2006.

\bibitem[BT09]{BT09}
S.~Bravyi and B.~Terhal.
\newblock Complexity of stoquastic frustration-free {H}amiltonians.
\newblock {\em SIAM Journal on Computing}, 39(4):1462, 2009.

\bibitem[CCD{\etalchar{+}}11]{CCDJZ11}
J.~Chen, X.~Chen, R.~Duan, Z.~Ji, and B.~Zeng.
\newblock No-go theorem for one-way quantum computing on naturally occurring
  two-level systems.
\newblock {\em Physical Review A}, 83:050301(R), 2011.

\bibitem[Coh93]{cohen93}
H.~Cohen.
\newblock {\em A Course in Computational Algebraic Number Theory}.
\newblock Graduate Texts in Mathematics. Springer, 1993.

\bibitem[Coo72]{C72}
S.~Cook.
\newblock The complexity of theorem proving procedures.
\newblock In {\em Proceedings of the 3rd ACM Symposium on Theory of Computing
  (STOC 1972)}, pages 151--158, 1972.

\bibitem[dB14]{Beaudrap14}
N.~de~Beaudrap.
\newblock Difficult instances of the counting problem for 2-quantum-sat are
  very atypical.
\newblock In {\em Proceeedings of TQC'14}, pages 118--140, 2014.
\newblock arXiv:1403.1588.

\bibitem[dBOE]{dBOE10}
N.~de~Beaudrap, T.~J. Osborne, and J.~Eisert.
\newblock Ground states of unfrustrated spin hamiltonians satisfy an area law.
\newblock {\em New Journal of Physics}, 12.

\bibitem[DP60]{DP60}
M.~Davis and H.~Putnam.
\newblock A computing procedure for quantification theory.
\newblock {\em Journal of the ACM}, 7(3):201, 1960.

\bibitem[EIS76]{EIS76}
S.~Even, A.~Itai, and A.~Shamir.
\newblock On the complexity of the time table and multi-commodity flow
  problems.
\newblock {\em SIAM Journal on Computing}, 5(4):691--703, 1976.

\bibitem[F\"07]{F07}
M.~F\"{u}rer.
\newblock Faster integer multiplication.
\newblock In {\em Proceedings of the 39th {ACM} {S}ymposium on the {T}heory of
  {C}omputing (STOC 2007)}, pages 55--67, 2007.

\bibitem[GN13]{GN13}
D.~Gosset and D.~Nagaj.
\newblock Quantum 3-{SAT} is {QMA1}-complete.
\newblock In {\em Proceedings of the 54th IEEE Symposium on Foundations of
  Computer Science (FOCS 2013)}, pages 756--765, 2013.

\bibitem[GZ11]{GZ11}
Oded Goldreich and David Zuckerman.
\newblock Another proof that $\mathcal{BPP}\subseteq \mathcal{PH}$ (and more).
\newblock In Oded Goldreich, editor, {\em Studies in Complexity and
  Cryptography. Miscellanea on the Interplay between Randomness and
  Computation}, volume 6650 of {\em Lecture Notes in Computer Science}, pages
  40--53. Springer Berlin Heidelberg, 2011.

\bibitem[Has06]{H06}
M.~B. Hastings.
\newblock Solving gapped {H}amiltonians locally.
\newblock {\em Physical Review B}, 73:085115, 2006.

\bibitem[JKNN12]{JKNN12}
S.~P. Jordan, H.~Kobayashi, D.~Nagaj, and H.~Nishimura.
\newblock Achieving perfect completeness in classical-witness quantum
  {M}erlin-{A}rthur proof systems.
\newblock {\em Quantum Information \& Computation}, 12(5 \& 6):461--471, 2012.

\bibitem[JWZ11]{JWZ11}
Z.~Ji, Z.~Wei, , and B.~Zeng.
\newblock Complete characterization of the ground space structure of two-body
  frustration-free hamiltonians for qubits.
\newblock {\em Physical Review A}, 84, 2011.

\bibitem[Kar72]{K72}
R.~Karp.
\newblock Reducibility among combinatorial problems.
\newblock In {\em Complexity of Computer Computations}, pages 85--103. New
  York: Plenum, 1972.

\bibitem[KKR06]{KKR06}
J.~Kempe, A.~Kitaev, and O.~Regev.
\newblock The complexity of the local {H}amiltonian problem.
\newblock {\em SIAM Journal on Computing}, 35(5):1070--1097, 2006.

\bibitem[Kro67]{K67}
M.~R. Krom.
\newblock The decision problem for a class of first-order formulas in which all
  disjunctions are binary.
\newblock {\em Zeitschrift f\"{u}r Mathematische Logik und Grundlagen der
  Mathematik}, 13:15--20, 1967.

\bibitem[KSV02]{KSV02}
A.~Kitaev, A.~Shen, and M.~Vyalyi.
\newblock {\em Classical and Quantum Computation}.
\newblock American Mathematical Society, 2002.

\bibitem[Lev73]{L73}
L.~Levin.
\newblock Universal search problems.
\newblock {\em Problems of Information Transmission}, 9(3):265--266, 1973.

\bibitem[LMSS10]{LMSS10}
C.~R. Laumann, R.~Moessner, A.~Scardicchio, and S.~L. Sondhi.
\newblock Phase transitions and random quantum satisfiability.
\newblock {\em Quantum Information \& Computation}, 10:1--15, 2010.

\bibitem[Pap91]{P91}
C.~Papadimitriou.
\newblock On selecting a satisfying truth assignment.
\newblock In {\em Proceedings of the 32nd Annual IEEE {S}ymposium on
  Foundations of Computing (FOCS 1991)}, pages 163--169, 1991.

\bibitem[Qui59]{Q59}
W.~V. Quine.
\newblock On cores and prime implicants of truth functions.
\newblock {\em The American Mathematical Monthly}, 66(5):755--760, 1959.

\bibitem[Tar72]{T72}
R.~E. Tarjan.
\newblock Depth fist search and linear graph algorithms.
\newblock {\em SIAM Journal on Computing}, pages 146--160, 1972.

\bibitem[vzGG03]{GG03}
J.~von~zur Gathen and J.~Gerhard.
\newblock {\em Modern Computer Algebra}.
\newblock Cambridge University Press, 2003.

\bibitem[ZF87]{ZF87}
S.~Zachos and M.~Furer.
\newblock Probabalistic quantifiers vs. distrustful adversaries.
\newblock In {\em Foundations of Software Technology and Theoretical Computer
  Science, 7th Conference}, pages 443--455, 1987.
\newblock Volume 287 of \emph{Lecture Notes in Computer Science}.

\end{thebibliography}

\appendix

\section{Details of runtime analysis}
\label{apx:runtime}
\subsection{Complexity of \solveq\ in terms of field operations}\label{sscn:fieldops}
We now show that \solveq\ requires $O(n+m)$ operations over $\complex$, for $n$ and $m$ the number of variables and clauses, respectively.
Below, we let $G'$ denote the interaction graph induced by vertices which remain unassigned at any fixed point in the algorithm's execution.

First, note that as we require the coefficients to be drawn from a number field $\F$ (\ie,~a finite-degree extension of the rational numbers $\Q$; see Section~\ref{scn:preliminaries}), they cannot encode any uncomputable numbers or difficult computations.
In particular, as we note in Section~\ref{sscn:bitcomplexity}, they can all be performed by deterministic polynomial time algorithms; and while \solveq\ may involve operations on coefficients outside of $\F$, these are all performed in an easily computed representation of a field extension.
We thus consider field-operation complexity a meaningful measure of the complexity of our algorithm (and that of Bravyi~\cite{B06}).

As each constraint $\Pi_{ij}$ acts on a constant number of qubits, elementary algebraic operations involving these constraints and their associated transfer matrices --- such as computing ranks in the preprocessing stage, determining if a constraint is a product constraint in Step~1, computing the eigenvectors of a cycle matrix in Step~2, multiplying vectors by transfer matrices in CRs, \etc\ --- require a constant number of field operations.

The preprocessing phase, described in Section~\ref{scn:alg}, is performed as follows.
We collect together the constraints on each pair of qubits $\{u,v\}$ in the form of a partial sum $S_t$ of the first $t$ constraints on $\{u,v\}$.
We compute $S_t$ by including only those projectors $\Pi_i$ for which ${\mathrm{rank}(\Pi_i {\;\!+\;\!} S_{i{-}1})} > {\mathrm{rank}(S_{i{-}1})}$: projectors $\Pi_i$ which fail this test do not restrict the joint states of $u$ and $v$ any further.
Having obtained a positive constraint operator $S_{uv} \in \lin{\cH \ox \cH}$ of maximum possible rank, we perform the appropriate reductions or reject outright, as appropriate.
Apart from the CRs which may be performed after assignments, the total work per constraint is $O(1)$, for a total contribution of $O(m)$.

Consider the total cost of all chain reactions, either directly performed or simulated, at any of the stages in \solveq.
We take care to simulate CRs in parallel, because it is possible for a given CR to traverse a constant fraction of the edges in $G'$ before detecting a conflict.
Following~\cite{EIS76}, we interleave the simulations of the two CRs, so that these simulations only last asymptotically as long as the CR which terminates without conflict.
As all of the edges traversed by the conflict-free CR are removed from $G'$, the total length of time of the simulated CRs is $O(m)$.

At Step~2 of \solveq, in which we attempt to find a discretising cycle, $G'$ is a disjoint union of irreducible instances of \tq.
Theorem~\ref{thm:DFSdiscret} (Cycle Discovery Theorem) guarantees that, if a discretizing cycle exists in the same component as a vertex $v \in V$, it can be found with a single depth-first search from $v$.
Moreover, as the transfer matrices for all remaining constraints are full rank in this case, a single conflict-free CR suffices to successfully assign satisfying states to all vertices in $G'$.
The cost of the DFS in the component of $v$ is asymptotically bounded by the number of edges in that component, except if $v$ happens to be isolated; thus its cost is $O(n+m)$.

The final stage of the algorithm, in which we renormalise the output vectors for each qubit in the original instance, can evidently be performed in $O(n)$ field operations.

\subsection{Complexity of \solveq\ in terms of bit operations}\label{sscn:bitcomplexity}

We now investigate the bit complexity of \solveq, for the case where the projectors are drawn from $\Q$.
(For product projectors with coefficients over an arbitrary number field, the analysis of Section~\ref{scn:runtime} is already sufficient, given that the transfer matrices are all $2 \x 2$ matrices.)

In a nutshell, the number of bit operations required is similar to the number of field operations, except with an increasing overhead as a result of performing arithmetic over the number field $\F$.
For instance: in a depth-first search to try to discover a discretising cycle, computing the transfer matrix along a path $P$ of length $\ell$ yields a matrix $\xT_P$ whose entries may require $O(\ell)$ bits to describe.
Calculations involving the determinants of such matrices may then involve multiplications of $O(\ell)$-bit integers.

\paragraph{Representing coefficients.} As mentioned earlier, we assume the algebraic coefficients in the input are given as elements of a finite-degree field extension $\F{\;\!:\;\!}\Q$.
As we mention in Section~\ref{scn:preliminaries}, we suppose that the number field $\F$ is presented as part of the input, by means of a minimal polynomial $p \in \Q[x]$ for which $\F \cong \Q[x] / p$, together with a specification of how $\F$ embeds into $\C$~\cite{cohen93}.
Let $\omega \in \F$ be the root of $p$ which corresponds to $x + (p) \in Q[x]/p$ with respect to this embedding.
As part of our requirement that the constraints in the input are all specified in $O(1)$ space, we require that the minimal polynomial and a description of the embedding of $\F$ into $\C$ (\eg,~by a disambiguating approximation to $\omega$ in $\Q[i] \subset \C$) are also specified with at most some constant number of bits.
Let $d = [\F{\!\;:\!\;}\Q] = \deg p$ represent the degree of the extension from which these coefficients are taken.
We represent coefficients $\alpha \in \F$ in a rationalised form
$
	\alpha = \tfrac{1}{\mu} c(\omega)
$
for $\mu > 0$ an integer, and $c(x)$ a polynomial with integer coefficients with degree less than $d$.

Given two coefficients $\alpha = {\tfrac{1}{\mu}c(\omega)}$ and $\alpha' = {\tfrac{1}{\mu'}c'(\omega)}$, we compute their sum as $\alpha+\alpha'={\tfrac{1}{\mu \mu'}\bigl[\mu' c(\omega)} + {\mu c'(\omega)\bigr]}$, and their product as $\alpha \alpha' = {\tfrac{1}{\mu \mu'} \bigl[ c(\omega) c'(\omega) \bigr]}$.
We may compute $\mu' c + \mu c' \in \Z[x]$ by taking the sums of the integer coefficients, represented as integer vectors.
The representation of $c(\omega) c'(\omega)$ may be evaluated as a formal multiplication of polynomials, resulting in a polynomial of degree at most $2d$: this is then reduced to a polynomial of degree less than $d$, by repeated application of the identity $p(\omega) \omega^t = 0$.
This eliminates all terms of the formal product where $\omega$ has degree higher than $d$, and contributes to the lower-order terms in $\omega$.
For each term $\omega^t$ for $t > d$, reducing it to a normal form by this substitution involves expansion to a polynomial with at most $d$ terms.
This requires $O(d^2) = O(1)$ arithmetic operations on integers.
Then the complexity of putting $\alpha + \alpha'$ or $\alpha \alpha'$ into normal form is dominated by the complexity of the integer arithmetic associated with the integer coefficients of $c$ and $c'$, and the integers $\mu$ and $\mu'$.

%
%


%

\paragraph{Coefficients of transfer operators.}
We represent the rank-1 projectors in an instance of \tq\ by any one of their rows $\bra{\eta}$, which also have coefficients drawn from $\F$.
We may then use $\bra{\eta}$ to compute (an operator proportional to) the transfer matrix corresponding to the same projector.
Specifically, for any constraint $\Pi_{uv} \propto \ketbra{\eta}{\eta}$, the corresponding transfer matrix $\xT_{u,v}$ may be computed (up to an insignificant scalar factor of $\sqrt 2$) as
\begin{equation}
	\label{eqn:computeXferMtx}
	{\xT_{u,v} = \Bigl[ \bra{\eta}_{u,v} \ox I_{v'} \Bigr] \Bigl[ I_u \ox \Bigl(\ket{01}-\ket{10}\Bigr)_{\!v,v'}  \Bigr].}
\end{equation}
A normal-form representation $\tfrac{1}{\mu} c(\omega)$ of the coefficients of these vectors may be computed in constant time, where $\mu$ and the coefficients of $c$ are themselves $O(1)$.

%

When computing transfer matrices during a CR, we perform a chain of length at most $O(n)$ additions and multiplications of coefficients in $\F$. Thus, even if the representation of constraints in the input each require only $O(1)$ space, the same is not true of the assignments $\ket{\psi_u}$ which may be given to the qubits, nor of the transfer matrices which may be computed to determine these assignments.
As we note in Section~\ref{scn:runtime}, the arithmetic operations in Step~1 may each be performed in $O(1)$ time, by avoiding such an increase in the representation-size of coefficients; but such an increase is unavoidable for iterated matrix products of non-singular transfer matrices, such as may occur in the preprocessing phase and in Step~2.
Iterated multiplication of $\ell$ matrices of shape $2\x2$, each of which are each expressible in $O(1)$ space, can be performed in time $O(\ell^2)$, with the result having coefficients expressible in space $O(\ell)$.

\paragraph{Operations in possibly-redundant quadratic field extensions.}
In solving the (quadratic) characteristic polynomial of a transfer operator, we may need to represent a square root $\sqrt{D}$ for $D \in \F$.
Absent an algorithm to efficiently compute the minimal polynomial of $\sqrt{D}$ over $\Q$, it may be difficult to determine whether $\sqrt{D}$ is an element of $\F$, which would be necessary to compute a representation of $\sqrt D$, either in $\F$ or in a field extension with easily computed normal forms.
Our approach is to simply ignore the question of representations with normal forms, and content ourselves with representations in which arithmetic operations and equality tests can be performed asymptotically as efficiently as over $\Q$.

For an input representation over $\F = \Q[\omega]$, let $\E$ be the minimal subfield of $\C$ which contains $\F$ and $\sqrt D$.
(In the case that $\sqrt D \in \F$, this field is $\F$ itself.)
Write $\omega' = \sqrt D$.
Consider a representation of the form $\alpha = \tfrac{1}{\mu}f(\omega, \omega')$, where $\mu \in \Z$ and $f \in \Z[x_1, x_2]$.
We may require that $f$ is at most linear in $x_2$, in which case we may write
\begin{equation}
	f(x_1, x_2) = f_0(x_1) + f_1(x_1) x_2.
\end{equation}
Write $a_0 = f_0(\omega)$ and $a_1 = f_1(\omega)$: then $\alpha = \tfrac{1}{\mu}(a_0 + a_1 \sqrt D)$.
\begin{itemize}
\item
	To test equality between $\alpha = \tfrac{1}{\mu}(a_0 + a_1 \sqrt D) \in \F$ and $\beta = \tfrac{1}{\nu}(b_0 + b_1 \sqrt D) \in \F$, we evaluate whether
	\begin{equation}
			\nu a_0 - \mu b_0 = (\mu b_1 - \nu a_1) \sqrt D,
	\end{equation}
	which is satisfied if either $\nu a_0 - \mu b_0 = \mu b_1 - \nu a_1 = 0$, or if $\sqrt{D} \in \F$ and in particular
	\begin{equation}
	\label{eqn:identifySqrtD}
	  \sqrt D = \frac{\nu a_0 - \mu b_0}{\mu b_1 - \nu a_1}\,.
	\end{equation}
	A rationalised form of the right-hand side of Eqn.~\eqref{eqn:identifySqrtD} can be evaluated with field operations in $\F$, with a bit-cost depending on the size of the representations of $\alpha$ and $\beta$.
	If the squares of the left- and right-hand side are equal, it suffices to verify that the left- and right-hand sides are in the same quadrant of the Argand plane (rather than opposite quadrants).
	Taking the principal square root, $\sqrt D$ will always have complex argument ranging over $(-\tfrac{\pi}{2},\tfrac{\pi}{2}]$; it then suffices to determine whether the complex argument of the expression on the right-hand side of Eqn.~\eqref{eqn:identifySqrtD} is in this range, or in the complement.
	As $\omega$ is specified in $O(1)$ bits at the input, we may evaluate which is the case in time $M(n)$, for $\alpha$, $\beta$ having representations of size $O(n)$.
	Given that the right-hand side is well-defined, equality fails in Eqn.~\eqref{eqn:identifySqrtD} if and only if $\alpha \ne \beta$.
\item
	To add $\alpha = \tfrac{1}{\mu}(a_0 + a_1 \sqrt D) \in \F$ and $\beta = \tfrac{1}{\nu}(b_0 + b_1 \sqrt D) \in \F$, we simply compute
	\begin{equation}
	  \alpha + \beta = \tfrac{1}{\mu\nu}\Bigl[(a_0 \nu + b_0 \mu) + (a_1 \nu + b_1 \mu) \sqrt D\Bigr]
	\end{equation}
	as we would for a genuine quadratic extension; similarly,
	\begin{equation}
	  \alpha \beta = \tfrac{1}{\mu \nu}\Bigl[(a_0b_0 + a_1 b_1 D) + (a_1 b_0 + a_0 b_1)\sqrt D\Bigr],
	\end{equation}
	where in each case we reduce these operations to arithmetic in $\F$, at a cost depending on the representation of $\alpha$ and $\beta$ in $\F$.
\end{itemize}
In particular: to perform equality tests or arithmetic operations between two elements in $\E$, which each require $O(n)$ bits to represent, involves a constant number of similar operations over $\F$.
Thus the cost of these operations is $O(M(n))$.
In the case that $\F = \Q$ rather than $\F = \Q[\omega]$ for $\omega \notin \Q$, similar arguments hold.

\paragraph{Complexity of computing eigenvalues and testing CRs.}
In Step~2, solving the (quadratic) characteristic polynomial of $\xT_C$ for a cycle starting and ending at a vertex $u$ requires evaluating the quadratic formula on the coefficients of $\xT_C$.
Evaluating the discriminant $D = (\Tr \xT_C)^2 - 4(\det \xT_C)$ of the characteristic polynomial can be performed in time $O(M(\ell))$, where $M(\ell)$ is the cost of the $\ell$-bit integer multiplications which are entailed in computing the normal forms involved.
Using the extension representation described above, we simply define $\E_C = \F[\sqrt D]$.
We then use a similar rationalised form representation for elements of $\E_C$ as for $\F$, and present coefficients in rationalised form over $\E_C$; arithmetic over $\E_C$ may now involve addition and multiplication of $\ell$-bit integers at each step.
Then the cost of each arithmetic operation or equality can only be bounded by $O(M(\ell))$.


The eigenvectors $\ket{\psi_{u,1}}$ and $\ket{\psi_{u,2}}$ of $\xT_C$ (which are not necessarily unit vectors) are elements of $\E_C$ with space complexity $O(\ell)$.
Application of the transfer matrices to $\ket{\psi_{u,j}}$ will again yield vectors of size $O(\ell)$, albeit representing vectors on $\E_C$.
Using the rationalised form described above, we may test whether two assignments to a vertex conflict simply by equality testing, which again may be performed in time $O(M(\ell))$.

Further CRs may give rise to different quadratic extensions $\E_{C'}$ of $\F$ which are in principle independent of one another.
However, by construction, each CR will involve only a quadratic extension of $\F$: By the Set-and-Forget Theorem (Theorem~\ref{thm:setandforget}), if a conflict-free CR is found from $u$, the state of any qubit not involved in the CR is not impacted by the assigned state on $u$.
Thus the extension $\E_C$ of one CR is not important to the operations performed in further CRs.

The potentially most expensive arithmetic operations of \solveq\ are then evaluating the quadratic formula, arithmetic on coefficients in the extensions $\E_C$, and equality testing to test for conflict between potentially different assignments to a qubit.
For a CR which visits $\ell$ distinct qubits, these each require time $M(\ell) \in \Omega(\ell)$. The complexity of \solveq\ may then be bounded by $O\bigl((n+m) M(n)\bigr)$, arising from performing $O(M(\ell))$-cost operations at most a constant number of times for each vertex and edge, where $\ell \in O(n)$.

\paragraph{Complexity of producing a normalised state as output.}
With no further asymptotic time complexity, we may compute unit vectors for the representation of the output.
If all vertices $v \in V$ have been assigned a vector $\ket{\psi_v}$, evaluating $\| \psi_v \| = \sqrt{\bracket{\psi_v}{\psi_v}}$ involves multiplication and square roots of integers of size $O(n)$.
Each vector $\ket{\psi_v}$ draws coefficients from either $\F$ or some quadratic extension $\E_C$; to express $\sqrt{\bracket{\psi_v}{\psi_v}}$ exactly may require a further (possibly trivial) extension, $\D_v := \E_C\bigl[\!\!\:\sqrt{\bracket{\psi_v}{\psi_v}}\,\bigr]$.
This final extension may be different for each vertex $v$, but can be represented by a minimal polynomial with $O(n)$ bits.
Evaluating $\tfrac{1}{\| \psi_v \|} \ket{\psi_v}$ then involves a final computation requiring time $O(M(n))$, so that this final normalisation step takes only $O(n M(n))$ time.

\paragraph{Size of the output.} The output consists of a list of vectors $\ket{\psi_v}$, each of which consists of two complex numbers drawn from some field extension $\D_v{\;\!:\;\!}\F$ of degree at most $4$.
The field $\F$ is determined from the input, and by hypothesis can be expressed in $O(1)$ bits; the extension to get $\D_v$ can be done using at most two minimal polynomials which can be expressed in $O(n)$ bits each.
The coefficients in $\D_v$ used by $\ket{\psi_v}$ are of the form $\alpha = \tfrac{1}{\mu} c$, where $c$ is an integer polynomial in the extension elements, where $|c|, |\mu| \in O(n)$.
Thus the space required to express each $\ket{\psi_v}$ is $O(n)$; the total space used by the output is then at most $O(n^2)$.

\end{document}